\documentclass{article}

\usepackage{arxiv}

\usepackage[utf8]{inputenc} 
\usepackage[T1]{fontenc}    
\usepackage{hyperref}       
\usepackage{url}            
\usepackage{booktabs}       
\usepackage{amsfonts}       
\usepackage{nicefrac}       
\usepackage{microtype}      
\usepackage{lipsum}

{}
{}
{}

\usepackage{amsmath}
\usepackage{amsfonts}
\usepackage{amssymb}
\usepackage{graphicx}
\usepackage{float}
\usepackage{caption}

\usepackage{bm}
\usepackage{breqn}

\usepackage{booktabs}
\usepackage{hhline}

\usepackage{algorithm}
\usepackage{algpseudocode}
\usepackage{multirow}

\newcommand\Algphase[1]{%
\vspace*{.5\baselineskip}
\Statex\hspace*{-\algorithmicindent}\textbf{#1}%
\vspace*{.5 \baselineskip}
}

\usepackage{array}

\usepackage{adjustbox}
\usepackage{makecell}

\usepackage{balance}
\usepackage{flushend}

\usepackage[caption=false,font=footnotesize]{subfig}

\usepackage{stfloats}

\usepackage{url}

\usepackage{amsthm}

\newtheorem{definition}{Definition}

\newtheorem{lem}{Lemma}

\graphicspath{{D:/Dropbox/Study/trial/measurement_mtx_extend/iet_revw_2/}}

\title{Entropy-Based Learning of Sensing Matrices}

\author{
  Gayatri Parthasarathy\thanks{This paper is a preprint of a paper submitted to IET Signal Processing journal. If accepted, the
copy of record will be available at the IET Digital Library} \\
  Electronics and Communication Engineering Department\\
  National Institute of Technology, Calicut\\
  Kerala-673601, India \\
  \texttt{gayatri\_p120094ec@nitc.ac.in} \\
   \And
 G.~ Abhilash \\
  Electronics and Communication Engineering Department\\
  National Institute of Technology, Calicut\\
  Kerala-673601, India \\
  \texttt{abhilash@nitc.ac.in} \\
}

\begin{document}
\maketitle
\begin{abstract}
This paper proposes a learning method to construct an efficient sensing (measurement) matrix, having orthogonal rows, for compressed sensing of a class of signals. The learning scheme identifies the sensing matrix by maximizing the entropy of measurement vectors. The bounds on the entropy of the measurement vector necessary for the unique recovery of a signal are also proposed. A comparison of the performance of the designed sensing matrix and the sensing matrices constructed using other existing methods is also presented. The simulation results on the recovery of synthetic, speech, and image signals, compressively sensed using the sensing matrix identified, shows an improvement in the accuracy of recovery. The reconstruction quality is better, using less number of measurements, than those measured using sensing matrices identified by other methods. 
\end{abstract}

\section{Introduction}

Compressed sensing (CS) aims at capturing signals, sparse in some domain, in a reduced set of measurements. Consider a signal $\bm{x}\in \mathbb{R}^N$, having representation $\bm{c}\in \mathbb{R}^N$ relative to a basis $\bm{\Psi}=\lbrace \bm{\psi_i}\rbrace_{i=1}^N$, such that $\bm{x}=\bm{\Psi c}$. The representation of the signal is $K$- sparse if $\Vert \bm{c}\Vert_0\leq K$ (which is the number of non-zero coefficients in the representation) for $K\ll N$. For the signal $\bm{x}$, the compressive sensing problem is given as
\begin{equation}
\bm{y}=\bm{\Phi x}=\bm{\Phi \Psi c}= \bm{Ac},
\label{measure}
\end{equation}
where the sensing matrix $\bm{\Phi}\in \mathbb{R}^{M\times N}$, $K<M\ll N$, and $\bm{A}=\bm{\Phi \Psi}$. 

The $\bm{\Phi}$ matrix should be constructed such that the matrix $\bm{A}$ captures maximum information from the sparse coefficients $\bm{c}$ for a known sparsifying basis $\bm{\Psi}$. The classical choice of the sensing matrix is a random matrix which simplifies the theoretical analysis \cite{1}\cite{2}. The scheme for realizing the random projections is the random demodulator proposed in \cite{3} and \cite{4}. 
In this paper, we address the problem of identifying an efficient sensing matrix.

In \cite{5}, Elad introduced a structured sensing matrix along with a method to construct it by reducing the mutual coherence $\mu (\bm{A})$ of the columns of the matrix $\bm{A}$. Subsequently, many techniques were proposed to construct the sensing matrix by reducing $\mu (\bm{A})$ \cite{7} -\cite{11}. In \cite{6} and \cite{8}, the sensing matrices were constructed by applying multidimensional scaling (MDS) on a sparsifying dictionary $\bm{\Psi}$. The methods employed in error control coding theory have also been used for the construction of sensing matrices \cite{12} -\cite{14}. Carson \textit{et al} \cite{24} proposed the construction of a projection matrix by maximizing the statistical Renyi entropy of the projections. It requires the knowledge of the statistical probability distribution of the signals. Baldassarre \textit{et al} \cite{15} proposed a method for constructing a measurement matrix for signals with structured sparsity by learning from the set of signals. The technique proposed in \cite{15} is a sub-sampling method which aims at capturing maximum energy of the structured sparse signals.

In  \cite{25} and \cite{26}, Hegde \textit{et al} proposed a learning method called NuMax to construct a sensing matrix such as to satisfy the restricted isometry property (RIP)\cite{30}.
\begin{equation}
(1-\delta)\Vert \bm{c}\Vert_2^2\leq \Vert \bm{A c}\Vert_2^2\leq (1+\delta)\Vert \bm{c}\Vert_2^2, \quad \text{where  } 0< \delta < 1.
\label{RIP}
\end{equation}

In this paper, we propose a learning method to construct an efficient measurement matrix (Section \ref{algo}), from a set of training signals belonging to a class of signals, by maximizing the Shannon entropy of the measurement vectors constrained to achieving RIP. The measurement matrix constructed using the proposed method has orthogonal rows. The probability distribution of the signals and the Shannon entropy used in this article (Section \ref{motivation}) are different from the statistical probability and entropy used in \cite{24}, respectively. A relation between the entropy $H(\bm{y})$ of the measurement vector $\bm{y}$, the number of measurements $M$, and the entropy $H(\bm{c})$ of the vector of representation coefficients of the signal (Section \ref{motivation}) is also proposed. The simulation results (Section \ref{results}) with synthetic and speech signals suggest that the method proposed is capable of constructing measurement matrices that give improved recovery from a reduced set of measurements. An advantage of the proposed method is that it works even for signals that do not have structured sparsity (Section \ref{synthetic}).  
\section{Motivation and Problem Formulation}
\label{motivation}
To motivate the entropy based measurement matrix design, we use the definitions of the probability distribution of the representation of a signal and the Shannon entropy of the representation of the signal as proposed in \cite{16} - \cite{35}:
\begin{definition} Let $\bm{\Psi}=\{\bm{\psi}_i\}_{i=1}^N$ be a basis of an $N$-dimensional space. Let $\bm{x}$ be a signal belonging to a class of signals $\bm{X}$ such that $\bm{x}= \sum_{i=1}^N{c_i\bm{\psi}_i}$,  where $\bm{c}=[c_1, c_2,...c_N]^T$ is the vector of representation coefficients of $\bm{x}$ relative to $\bm{\Psi}$. The probability distribution of the representation of the signal $\bm{x}$ relative to $\bm{\Psi}$ is $\bm{p}=\lbrace p_i \rbrace_{i=1}^N$, where $p_i=|c_i|^2/\Vert \bm{c} \Vert_2^2$.
\end{definition} 
The entropy of representation conditioned to $\bm{\Psi}$ is the amount of information left in $\bm{x}$ when the representation basis $\bm{\Psi}$ is known \cite{28}. This conditional entropy $H(\bm{x}\mid \bm{\Psi})$ is defined as follows.
\begin{definition}
The Shannon entropy of the signal $\bm{x}$ with respect to basis $\bm{\Psi}$ is given by 
\begin{equation}
H(\bm{x}\mid\bm{\Psi})=-\sum_{i=1}^N p_i \ln(p_i)=-\sum_{i=1}^N \frac{|c_i|^2}{\Vert \bm{c} \Vert_2^2}\ln\left(\frac{|c_i|^2}{\Vert \bm{c} \Vert_2^2}\right).
\end{equation}
\label{ent_def}
\end{definition}
\textit{Remarks}: By definition, $0\times \ln(0)=0$ \cite{29}. The terms Shannon entropy and entropy are used interchangeably in this article.

Since the entropy of representation with respect to the basis $\bm{\Psi}$ depends only on the probability distribution of representation $\bm{p}$, which depends on the coefficients of representation $\bm{c}$, we can conclude $H(\bm{x}\mid\bm{\Psi})=H(\bm{c})$, with $H(\bm{c})$ being the \textit{entropy of the representation} of $\bm{x}$ relative to $\bm{\Psi}$. 
The entropy of representation of the reduced set of measurements $\bm{y}$ with respect to some basis of the range space of $\bm{A}$ is termed as the \textit{entropy of the set of measurements} or \textit{entropy of the measurement vector}. For simplicity, we consider the representation basis of the measurements to be the standard ordered basis of an $M$-dimensional space. Hence the entropy of the measurement vector is 
\begin{equation}
H(\bm{y})=-\sum_{i=1}^M \frac{|y_i|^2}{\Vert \bm{y} \Vert_2^2}\ln\left(\frac{|y_i|^2}{\Vert \bm{y} \Vert_2^2}\right),
\label{measur_ent}
\end{equation}
 where $\bm{y}$ is given by (\ref{measure}) and $y_i$'s are the representation coefficients relative to the standard ordered basis of the reduced space of measurements which is the range space of $\bm{A}$.  

From the definition of the entropy of representation $H(\bm{c})$, it can be seen that the more concentrated the probability distribution of the representation, the lower is the entropy. Since the probability of representation is directly proportional to the magnitude of the representation coefficient, it can be argued that the lower the entropy $H(\bm{c})$, the higher is the compressibility of the representation \cite{31}. To quantify the compressibility of the signal with respect to a basis, we introduce theoretical dimension as the entropy-based measure of sparsity. Given the entropy of representation $H(\bm{c})$, the theoretical dimension $n_{th,c}^{\Psi}$ of the representation of a signal in the basis $\bm{\Psi}$ is given by \cite{17} \cite{31} \cite{18}-\cite{35}
\begin{equation}
n_{th,c}^{\Psi}=\lceil \exp(H(\bm{c}))\rceil.
\label{th_dim}
\end{equation}
The theoretical dimension takes values $1 \leq n_{th,c}^{\Psi} \leq N$, where $N$ is the total number of basis vectors in the representation basis. As entropy decreases, the theoretical dimension also decreases leading to a compressible representation of the signal. The theoretical dimension specifies the number of basis vectors required to represent a compressible signal without unduly degrading the signal quality. Experimental results of \cite{31} show that the theoretical dimension gives the number of basis vectors required to capture at least $90\%$ of the signal energy. The various criteria that a measure of sparsity must satisfy are discussed in \cite{41}. Meena and Abhilash \cite{35} discuss the various criteria, mentioned in \cite{41}, that the theoretical dimension based measure of sparsity satisfies. 


 For a compressible signal, the $l_0$-sparsity is achieved by restricting the representation to those coefficients which carry the maximum energy of the signal, such that the signal quality is not degraded unduly. Hence, the theoretical dimension of the signal representation $\bm{c}$ can be approximated as the $l_0$- sparsity $K$ of the representation of the signal, that is $K\approx n_{th,c}^{\Psi}$. According to the definition of theoretical dimension, the lower the entropy, the higher is the sparsity or the lower is the value of $K$. The theoretical dimension of $\bm{y}$ with respect to the standard ordered basis of an $M$ dimensional space  is $n_{th,y}^I=\lceil \exp(H(\bm{y}))\rceil=M_{eff}$. The quantity $M_{eff}$ is the effective number of coefficients, in the representation of the measurement vector $\bm{y}$ relative to the standard ordered basis, that captures at least $90\%$ of the signal energy. The remaining $M-M_{eff}$ number of coefficients in the representation of $\bm{y}$ carries insignificant amount of information.
 
According to the theory of CS, for any measurement vector $\bm{y} \in R^M$ there exists at most one $K$-sparse signal $\bm{c} \in R^N$, such that $\bm{y}=\bm{Ac}$, if and only if $\text{spark}(\bm{A}) > 2K$ \cite{30}, where $\text{spark}(\bm{A})$ is the smallest number of columns of $\bm{A}$ that are linearly dependent. The value of spark($\bm{A}$) lies in the range $[2,M+1]$, where $M$ is the dimension of the range space of $\bm{A}$ \cite{30}. Hence the requirement $M\geq 2K$ holds good. Since $M_{eff}$ gives the effective number of measurements, ideally, the necessary condition for unique recovery would be $M_{eff} \geq 2K$ (See Appendix A). Since we use the entropy of the measurement vector $\bm{y}$ to construct the measurement matrix, we identify the bounds on the entropy of the measurement vector, for a fixed $M$. These bounds are necessary, but not sufficient, for achieving unambiguous measurements for unique recovery. These bounds are presented in Lemma 1.  
 
\begin{lem}
If $\bm{y}$ is a measurement vector consisting of a set of $M\geq 2K$ measurements of a non-zero compressible signal, then the entropy $H(\bm{y})$ of the measurement vector, relative to the standard ordered basis of an $M$-dimensional space, satisfies
\begin{align}
H(\bm{c})\leq \ln(K)<H(\bm{y})\leq \ln(M)
\label{lem1} .
\end{align}
where $H(\bm{c})$ is the entropy of the representation of a compressible signal relative to a sparsifying basis and the approximate $l_0$ sparsity of the signal is $K \approx \lceil \exp(H(\bm{c}))\rceil$.
\end{lem} 
\begin{proof}
If $H(\bm{y})\leq \ln(K)$, then $\exp(H(\bm{y})) \leq K$. We know that $M_{eff}=\lceil \exp(H(\bm{y}))\rceil$. Hence $M_{eff}=\exp(H(\bm{y}))+\epsilon$, with $0\leq \epsilon <1$. Thus,
\begin{align}
\exp(H(\bm{y})) \leq K \Rightarrow M_{eff}\leq K+\epsilon.
\end{align}
Therefore, the case $H(\bm{y})\leq \ln(K)$ violates the requirement $M_{eff} \geq 2K$. Hence by contradiction, $\ln(K)<H(\bm{y})$. Since by (\ref{th_dim}), $H(\bm{c})$ is at the most $\ln(K)$, the lowest bound in (\ref{lem1}) is true.    

Given the number of measurements $M$, the entropy $H(\bm{y})$ of the measurement $\bm{y}$ is at the most $\ln(M)$ (by (\ref{measur_ent})); hence the upper bound in (\ref{lem1}) also holds. Consider two signals $\bm{x_1}$ and $\bm{x_2}$ $(\bm{x_1}\neq \bm{x_2})$, with measurements $\bm{y_1}$ and $\bm{y_2}$ such that $H(\bm{y_1})=H(\bm{y_2})=\ln(M)$, then by Definitions 1 and 2, the probability distribution of the $j$-th measurement vector $\bm{y_j}$ with respect to the standard ordered basis of a vector space of dimension $M$ is $p_{ij}=\frac{\mid y_{ij}\mid ^2}{\Vert \bm{y_j}\Vert_2^2}=1/M$, for $i=1,2...M$. The equality of the probability distribution does not imply equality of the measurement vectors. That is, $\frac{\mid y_{i1}\mid ^2}{\Vert \bm{y_1}\Vert_2^2}=\frac{\mid y_{i2}\mid ^2}{\Vert \bm{y_2}\Vert_2^2}$ does not imply $y_{i1}=y_{i2}$. Hence the equality $H(\bm{y})=\ln(M)$ does not affect the uniqueness of the measurements.
\end{proof}

Let $\mathcal{N}(\bm{A})$ represent the null space of the matrix $\bm{A}$. If $\bm{c}\in \mathcal{N}(\bm{A})$ (where $\bm{c}$ is the $K$-sparse approximation of the representation of a non-zero compressible signal), then the vector of the reduced set of measurements is $\bm{y}=\bm{Ac}=\bm{\theta}$, where $\bm{\theta}\in \mathbb{R}^M$ is the zero vector. Since the probability distribution of the measurements is defined as $p_i= \left\lbrace \frac{\vert y_i\vert^2}{\Vert \bm{y}\Vert_2^2}\right\rbrace_{i=1}^M$, the entropy $H(\bm{y})$ of $\bm{y}$ is undefined for $\bm{c}\in \mathcal{N}(\bm{A})$, as $p_i$'s are undefined. Hence, $\bm{c}\in \mathcal{N}(\bm{A})$ is undesirable. The theory of CS also discusses that the measurement operator $\bm{A}$ can uniquely represent all non-zero $K$-sparse signals if and only if no non-zero $2K$-sparse signal lies in $\mathcal{N}(\bm{A})$ \cite{30}. The following Lemma shows that maximization of the entropy $H(\bm{y})$ of measurements implies that no non-zero $K$-sparse signal lies in $\mathcal{N}(\bm{A})$. Ideally, if no non-zero $K$-sparse signal falls in $\mathcal{N}(\bm{A})$, then no non-zero $2K$-sparse signal will belong to $\mathcal{N}(\bm{A})$.   

\begin{lem}
If $\bm{A}$ is a measurement matrix having rank $M$ and $\bm{y}\in \mathbb{R}^M$ is a vector of measurements of the $K$-sparse approximation of a non-zero compressible signal with $M\geq 2K$, then the maximization of the entropy $H(\bm{y})$ of $\bm{y}$ implies that no $K$-sparse signal falls in the null space $\mathcal{N}(\bm{A})$ of $\bm{A}$. 
\end{lem}
\begin{proof}
Since the rank of $\bm{A}$ is $M$, $\text{spark}(\bm{A})\in [2,M+1]$. Since $M\geq 2K$, $\text{spark}(\bm{A})>2K$, implying the existence of at most one $K$-sparse signal $\bm{c}\in \mathbb{R}^N$ such that $\bm{y}=\bm{Ac}$.

Let $\bm{c}$ represent the vector of sparse representation coefficients of a non-zero compressible signal approximated to be $K$-sparse. If $\bm{c}\in \mathcal{N}(\bm{A})$, then the vector of the reduced set of measurements is $\bm{y}=\bm{Ac}=\bm{\theta}$. If $\Vert \bm{y}\Vert_0=1$, then $\bm{c}\not\in \mathcal{N}(\bm{A})$ and $H(\bm{y})=0$. But, by Lemma 1, for $K=1$, $0< H(\bm{y})\leq \ln(M)$. Hence, $H(\bm{y})=0$ indicates that the measurements are incomplete. Therefore, to capture maximum information of the signal into $M$ measurements, the entropy $H(\bm{y})$ of $\bm{y}$ should be maximized. An entropy maximized non-zero measurement $\bm{y}$ implies that $\bm{c}\not\in \mathcal{N}(\bm{A})$.  
%
\end{proof}
Lemma 2 maintains that maximization of the entropy $H(\bm{y})$ guarantees the null space property of $\bm{A}$. However, it does not claim that the maximization of entropy $H(\bm{y})$ would make the set of measurements complete. But, for unique recovery the set of measurements should capture maximum information contained in the signal. In particular, if the number of measurements is as small as possible, then the compression gain is the highest. Lemma 3 establishes how maximization of the entropy $H(\bm{y})$ enables unique recovery with the least possible number of measurements.
\begin{lem}\label{lemma3}
The maximization of the entropy $H(\bm{y})$ of the set of measurements of a compressible signal leads to unique recovery with a reduced number of measurements $M$ close to $2K$.
\end{lem}
\begin{proof}
$M_{eff}$ is the number of non-zero coefficients of the measurement vector in the standard ordered basis, that capture at least $90\%$ of the energy. By definition $M_{eff}\leq M$. Let $M_{eff}=M-\nu$ with $0\leq \nu$. As mentioned earlier, the necessary condition for unique recovery is $M \geq 2K$ \cite{30}. If $M=2K$ then,
\begin{equation}
M_{eff}=2K-\nu .
\label{lem3_1}
\end{equation}
Since $M_{eff}$ gives the effective number of measurements, ideally the necessary condition for unique recovery would be $M_{eff} \geq 2K$. This is not satisfied by (\ref{lem3_1}). Hence the ideal lower bound on $M$ would be $2K+\nu$.

Since $M_{eff}=\lceil \exp(H(\bm{y}))\rceil$, increasing $H(\bm{y})$ leads to increase in $M_{eff}$, thus
\begin{equation}
M-M_{eff}\rightarrow 0 \Rightarrow \nu \rightarrow 0.
\label{lem3_2}
\end{equation} 

From (\ref{lem3_1}) and (\ref{lem3_2}) we see that increasing $H(\bm{y})$ results in reducing $\nu$ and hence $M_{eff}\approx M$. Hence, as $\nu$ decreases, the lower bound $2K+\nu$ on $M$ comes close to $2K$. Thus, maximization of $H(\bm{y})$ leads to unique recovery with a reduced number of measurements $M$ close to $2K$.  
\end{proof}

Based on Lemmas 1, 2 and 3, it can be concluded that the matrix $\bm{A}$ that maximizes $H(\bm{y})$ would capture maximum information from the vector of coefficients $\bm{c}$ into a reduced set of measurements $\bm{y}$. It also implies that the smallest required number of measurements $M$ could be as small as $2K$. Hence, we propose a learning scheme for identifying $\bm{A}$, and thus $\bm{\Phi}$ for a class of signals such that $\bm{\Phi}$ maximizes the entropy $H(\bm{y})$ of the measurement vector $\bm{y}$.  
\section{Entropy Maximizing Sensing (EMS) Matrix Design}  
\label{algo}
We propose a two-stage learning procedure to identify an efficient measurement matrix $\bm{\Phi}$ for a class of compressible signals. The learning method is motivated by the two stage dictionary/transform learning algorithms \cite{31}\cite{36} \cite{21}. The dictionary learning also finds application in Blind Compressive Sensing (BCS) \cite{40}, where the measurement matrix $\bm{A}=\bm{\Phi\Psi}$ (with $\bm{\Psi}$ unknown) is learned using the dictionary learning approach. In this paper, we do not consider the BCS framework.  

The first stage of the proposed algorithm finds a set of entropy maximized measurements of the signals in the training set $\bm{X}$. The second stage tries to learn a $\bm{\Phi}$ based on the desired measurements, sparsifying basis, and the set of training signals. The two stages of the algorithm are alternately performed for a fixed number of iterations.
\subsection{Stage I}
Consider the matrix of $N$ dimensional training signals $\bm{X}\in \mathbb{R}^{N \times L}$, where each column is a signal from the training set and $L$ is the number of signals in the training set. Let $\bm{\Psi}\in \mathbb{R}^{N\times N}$ be a known orthonormal sparsifying basis (with the basis vectors $\lbrace \bm{\psi}_i\rbrace_{i=1}^N$ arranged as its columns) such that $\bm{X}=\bm{\Psi C}$, where $\bm{C}\in\mathbb{R}^{N\times L}$ is the matrix of the representation coefficients of the signals in $\bm{X}$ with respect to $\bm{\Psi}$. That is, $\bm{C}_p$ (the p-th column of $\bm{C}$) is the representation of the p-th signal $\bm{X}_p$ with respect to $\bm{\Psi}$. Given $\bm{\Psi}$ and an initial measurement matrix $\bm{\Phi}\in \mathbb{R}^{M\times N}$, we need to find the reduced set of measurements $\bm{Y}\in\mathbb{R}^{M\times L}$, of the training signals in $\bm{X}$, where $\bm{Y}=\bm{\Phi X}=\bm{\Phi \Psi C}=\bm{ AC}$, such that the entropy of the set of measurements of each signal in $\bm{X}$ is maximized. Since for the $p$-th signal $\bm{X_p}$, the measurement vector $\bm{Y_p}=\bm{AC_p}$ is independent of the measurement of the $q$-th signal $\bm{Y_q}=\bm{AC_q}$ $(p\neq q)$, we can update the measurements considering all the signals separately. Thus, the problem can be formulated as 
\begin{equation}
\widehat{\bm{Y}}_j=\text{arg}\max_{\bm{Y}_j} H(\bm{Y}_j),
\end{equation} 
where $\bm{Y}_j$ is the $j$-th column of $\bm{Y}$, and $\widehat{\bm{Y}}_j$, the $j$-th column of $\widehat{\bm{Y}}$, is the measurement vector of the $j$-th signal $\bm{X}_j$ having maximized entropy. Using Definition \ref{ent_def}, we can rewrite the problem as
\begin{equation}
\widehat{\bm{Y}}_j=\text{arg} \max_{\bm{Y}_j} \sum_{i=1}^M -\frac{|y_{ij}|^2}{\Vert \bm{Y}_j \Vert_2^2}\ln\left(\frac{|y_{ij}|^2}{\Vert \bm{Y}_j \Vert_2^2}\right) ,
\label{max_ent1}
\end{equation}
where $\bm{Y}_j=\lbrace y_{ij}\rbrace_{i=1}^M$. The RIP in (\ref{RIP}) can be restated as 
\begin{equation}
(1-\delta)\leq \frac{\Vert \bm{A C}_j\Vert_2^2}{\Vert \bm{C}_j\Vert_2^2} \leq (1+\delta).
\vspace{-2.5pt}
\end{equation}
Thus, we can reformulate the problem (\ref{max_ent1}), to satisfy the RIP, as
\begin{align}
\widehat{\bm{Y}}_j=\text{arg} \max_{\bm{Y}_j} \sum_{i=1}^M -\frac{|y_{ij}|^2}{\Vert \bm{Y}_j \Vert_2^2}\ln\left(\frac{|y_{ij}|^2}{\Vert \bm{Y}_j \Vert_2^2}\right) \nonumber\\
\text{subject to  } \left\lbrace \left(\frac{\Vert \bm{Y}_j\Vert_2}{\Vert \bm{C}_j\Vert_2}\right)^2-1\right\rbrace^2\leq\delta^2 .
\label{max_ent2}
\end{align} 
The constrained problem in (\ref{max_ent2}) can be made unconstrained by using the penalty method. To incorporate the penalty, we convert the maximization problem to a minimization problem. Hence the stage I solution is given as
\begin{dmath}
\widehat{Y}_j=\text{arg} \min_{\bm{Y}_j} \left\lbrace\sum_{i=1}^M \frac{|y_{ij}|^2}{\Vert \bm{Y}_j \Vert_2^2}\ln\left(\frac{|y_{ij}|^2}{\Vert \bm{Y}_j \Vert_2^2}\right)+
\alpha \left\vert \left(\left\lbrace \left(\Vert \bm{Y}_j\Vert_2/\Vert \bm{C}_j\Vert_2\right)^2-1\right\rbrace^2\right) -\delta^2\right\vert \right\rbrace .
\label{max_ent3}
\end{dmath} 
Since the absolute value function of the penalty term is non-differentiable at zero, we approximate it using the relaxation $\vert \bm{z}\vert\approx\vert \bm{z}\vert^{\zeta}=\sqrt{\bm{z}^*\bm{z}+\zeta}$ where $\zeta$ is taken to be $10^{-15}$ \cite{32}. The value of $\delta$ is chosen depending on the desired RIP constant. The bounds on the value of $\delta$ for signals compressible in a dictionary are discussed in \cite{42}. We have arbitrarily chosen $\delta=0.1$ for the experiments in the paper. The problem (\ref{max_ent3}) can be solved using any optimization algorithm with the initial vector as $\bm{Y}_j=\bm{\Phi\Psi C}_j$.

\subsection{Stage II}
Stage I gives the desired measurements $\widehat{\bm{Y}}$ that maximizes the entropy of the measurement vectors. In stage II, we identify the matrix $\bm{A}$ and hence $\bm{\Phi}$ that would lead to $\widehat{\bm{Y}}$. The problem can be formulated as 
\begin{equation}
\widehat{\bm{A}}=\text{arg}\min_{\bm{A}} \Vert \widehat{\bm{Y}} -\bm{A C}\Vert_F^2.
\label{phi_update1}
\end{equation}
where the matrix $\bm{A}$ has orthogonal rows and $\Vert . \Vert_F$ represents the Frobenius norm. To solve the problem in (\ref{phi_update1}), we use the orthogonal Procrustes method. The orthogonal Procrustes problem  \cite{21} \cite{20} \cite{33} is: 
\begin{align}
\widehat{\bm{R}}=\text{arg}\min_{\bm{R}}\Vert \bm{RP}-\bm{Q} \Vert_F^2 \qquad \text{s.t.   } \bm{RR}^T=I,
\label{pro}
\end{align}
where $\widehat{\bm{R}}$ is an orthogonal square matrix to be found such that $\widehat{\bm{R}}$ acts on a matrix $\bm{P}$ to result in the matrix $\bm{Q}$. Considering the singular value decomposition of $\bm{PQ}^T$ as $\widetilde{\bm{U}}\widetilde{\bm{\Delta}} \widetilde{\bm{V}}^T$, the solution to (\ref{pro}) is $\bm{R}=\widetilde{\bm{V}}\widetilde{\bm{U}}^T$.

Since the matrix $\widehat{\bm{A}}$ in (\ref{phi_update1}) is not a square matrix, the solution to the Orthogonal Procrustes Problem has to be altered. If the singular value decomposition of $\bm{C}\widehat{\bm{Y}}^T$ is given by $\bm{U\Delta V}^T$, where $\bm{U}\in \mathbb{R}^{N\times N}$ and $\bm{V}\in \mathbb{R}^{M\times M}$, we propose that the solution to the problem (\ref{phi_update1}) is given by $\widehat{\bm{A}}=\bm{VU_M}^T$, where $\bm{U_M}$ contains the first $M$ columns of $\bm{U}$ which correspond to the largest $M$ singular values of $\bm{C\widehat{Y}}^T$ (see Appendix B). The desired measurement matrix $\widehat{\bm{\Phi}}$ can be obtained from $\widehat{\bm{A}}$ as
\begin{equation}
\widehat{\bm{\Phi}}=\widehat{\bm{A}}\bm{\Psi}^{-1},
\end{equation}
which holds good because $\bm{\Psi}$ is a well conditioned matrix of the representation basis for the class of signals considered.
\begin{figure}[t!]
\centering
\subfloat[][]{\includegraphics[width=4.5cm, height=3.5cm]{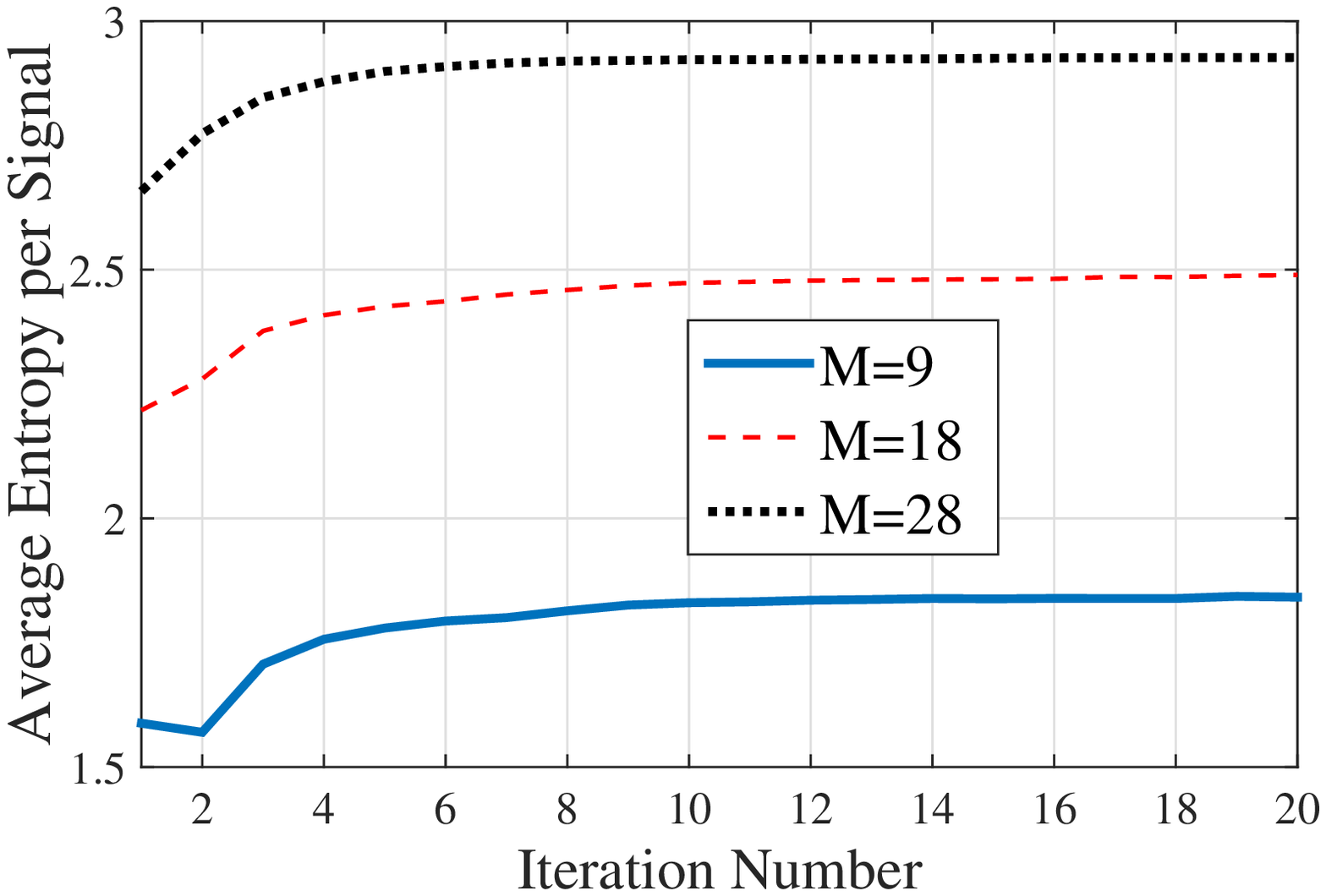}} \subfloat[][]{\includegraphics[width=4.5cm, height=3.5cm]{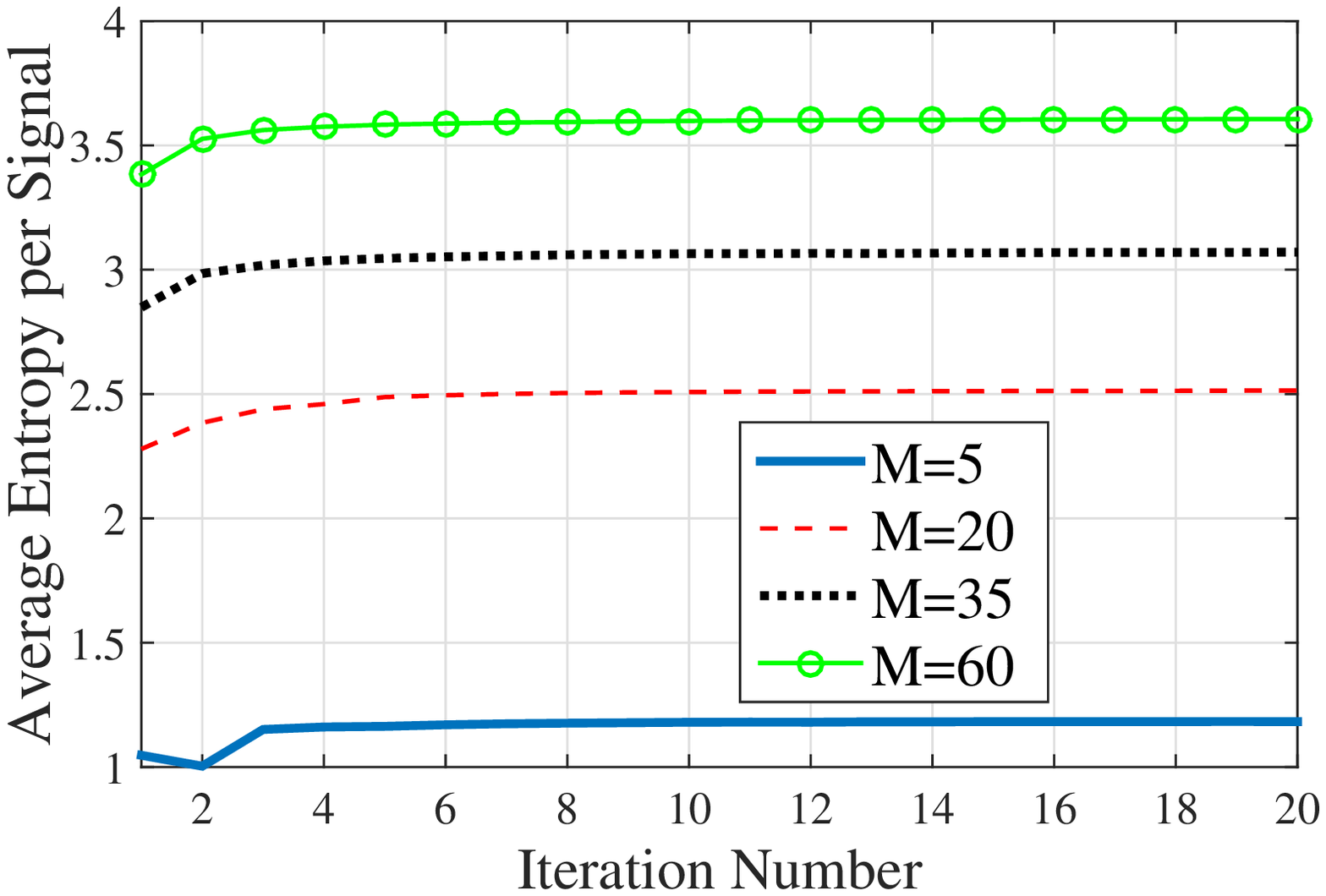}}
\caption{Saturation of average entropy per signal for (a) synthetic signals and (b) speech signals.}
\label{convergence}
\end{figure}
\begin{algorithm}[t!]
\caption{Entropy maximizing sensing(EMS) matrix design}
\begin{algorithmic}[1]
\Require Training set $\bm{X}_{N\times L}$, Initial measurement matrix $\widehat{\bm{\Phi}}_{M\times N}^{(0)}$, Sparsifying transform $\bm{\Psi}_{N\times N}$ 
\Ensure Measurement matrix $\widehat{\bm{\Phi}}_{M\times N}$

\State Initialize $\alpha=1$ and $\delta=0.1$
\State $\bm{C}=\bm{\Psi}^{-1} \bm{X}$
 \For{$k=1$ to $n$} 
\Algphase{Stage I}
 \For{$	j=1$ to $L$}
 \State Find the desired measurements $\widehat{\bm{Y}}_j^{(k)}$ starting with the initial measurement $\bm{Y}_j=\widehat{\bm{\Phi}}^{(k-1)}\bm{\Psi C}_j$ 
 
 \begin{align}
\widehat{Y}_j=\text{arg} \min_{\bm{Y}_j} \sum_{i=1}^M \frac{|y_{ij}|^2}{\Vert \bm{Y}_j \Vert_2^2}\ln\left(\frac{|y_{ij}|^2}{\Vert \bm{Y}_j \Vert_2^2}\right)+\nonumber\\
\alpha \left\vert \left(\left\lbrace \left(\Vert \bm{Y}_j\Vert_2/\Vert \bm{C}_j\Vert_2\right)^2-1\right\rbrace^2\right) -\delta^2\right\vert .\nonumber
\end{align} 

 \EndFor
\Algphase{Stage II}
 \State $\bm{E}=\bm{C}(\widehat{\bm{Y}}^{(k)})^T$
 \State $\bm{E}=\bm{U\Delta V}^T$
 \State Obtain $\bm{U_M}$ as the first $M$ columns of $\bm{U}$. 
 \State $\widehat{\bm{A}}^{(k)}=\bm{VU_M}^T$.
 \State Find the measurement matrix $\widehat{\bm{\Phi}}^{(k)}=\widehat{\bm{A}}^{(k)}\bm{\Psi}^{-1}$
 \EndFor
\end{algorithmic}
\end{algorithm}

Fig. \ref{convergence} (a) and Fig. \ref{convergence}(b) show the variation of the average entropy per signal as the iteration progresses for synthetic and speech signals, respectively; the variation saturates after a finite number of iterations indicating the convergence of the algorithm empirically. The figures were generated by setting $\alpha=1$ and $\delta=0.1$ in the algorithm.
\begin{figure}[ht!]
\centering
\includegraphics[width=6cm, height=4.5cm]{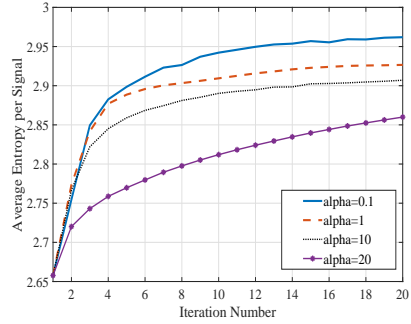}
\caption{Saturation of average entropy per signal (10-sparse synthetic signal) with $M=28$ and for different values of $\alpha$.}
\label{alpha_ver_reco}
\end{figure}
\begin{figure*}[ht!]
\centering
\subfloat[][]{\includegraphics[width=5.5cm, height=4cm]{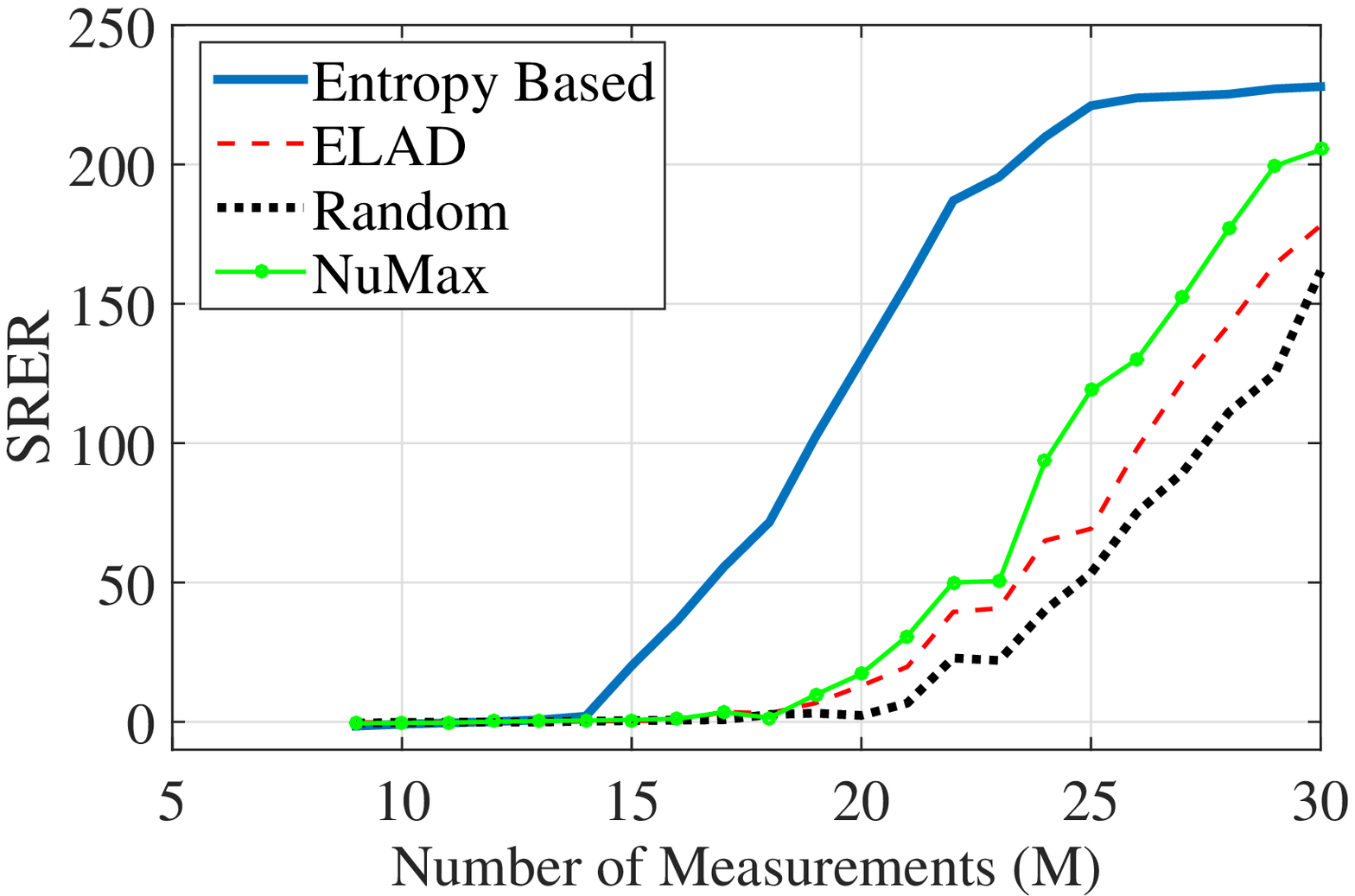}}\quad\subfloat[][]{\includegraphics[width=5.5cm, height=4cm]{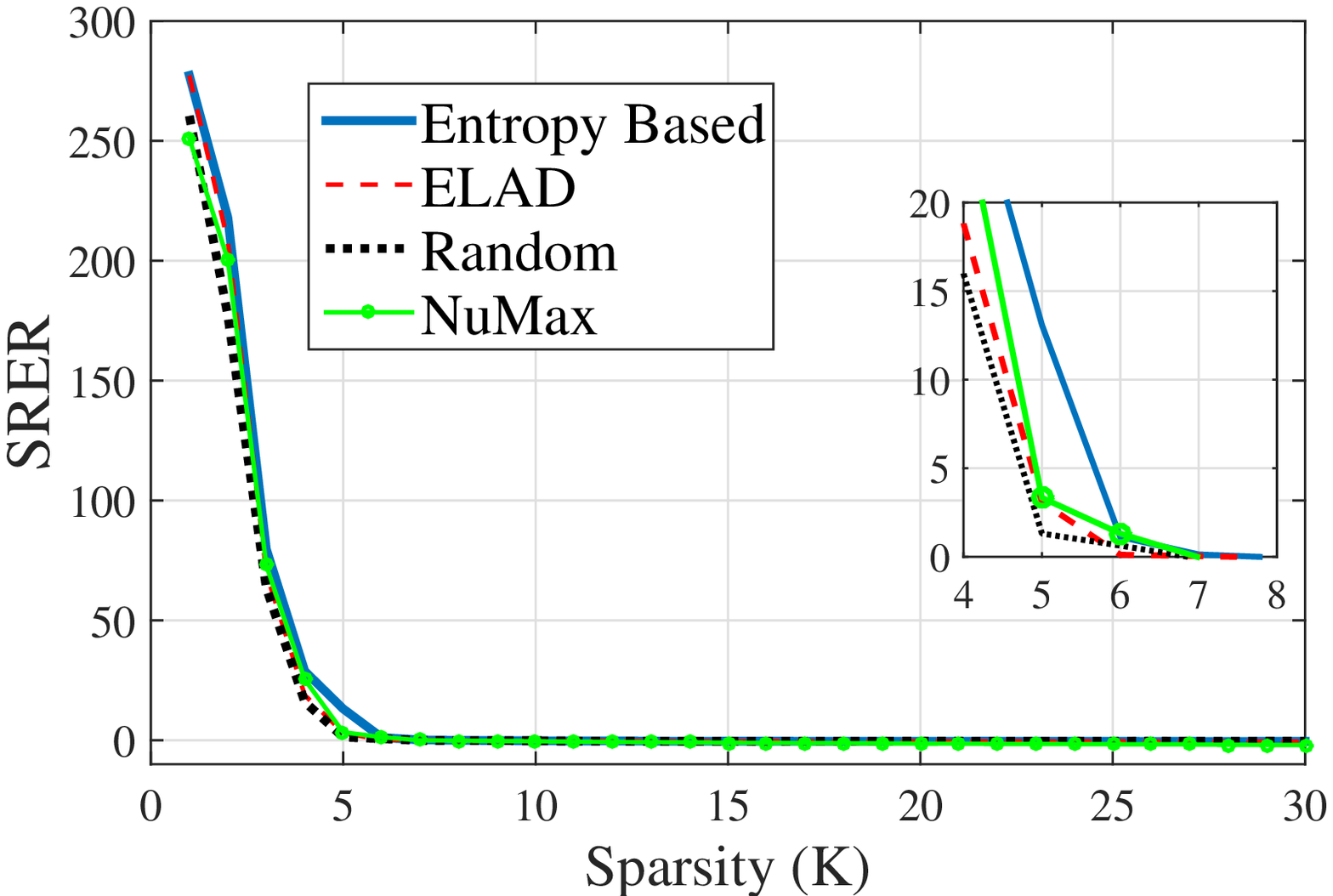}}\\ \subfloat[][]{\includegraphics[width=5.5cm, height=4cm]{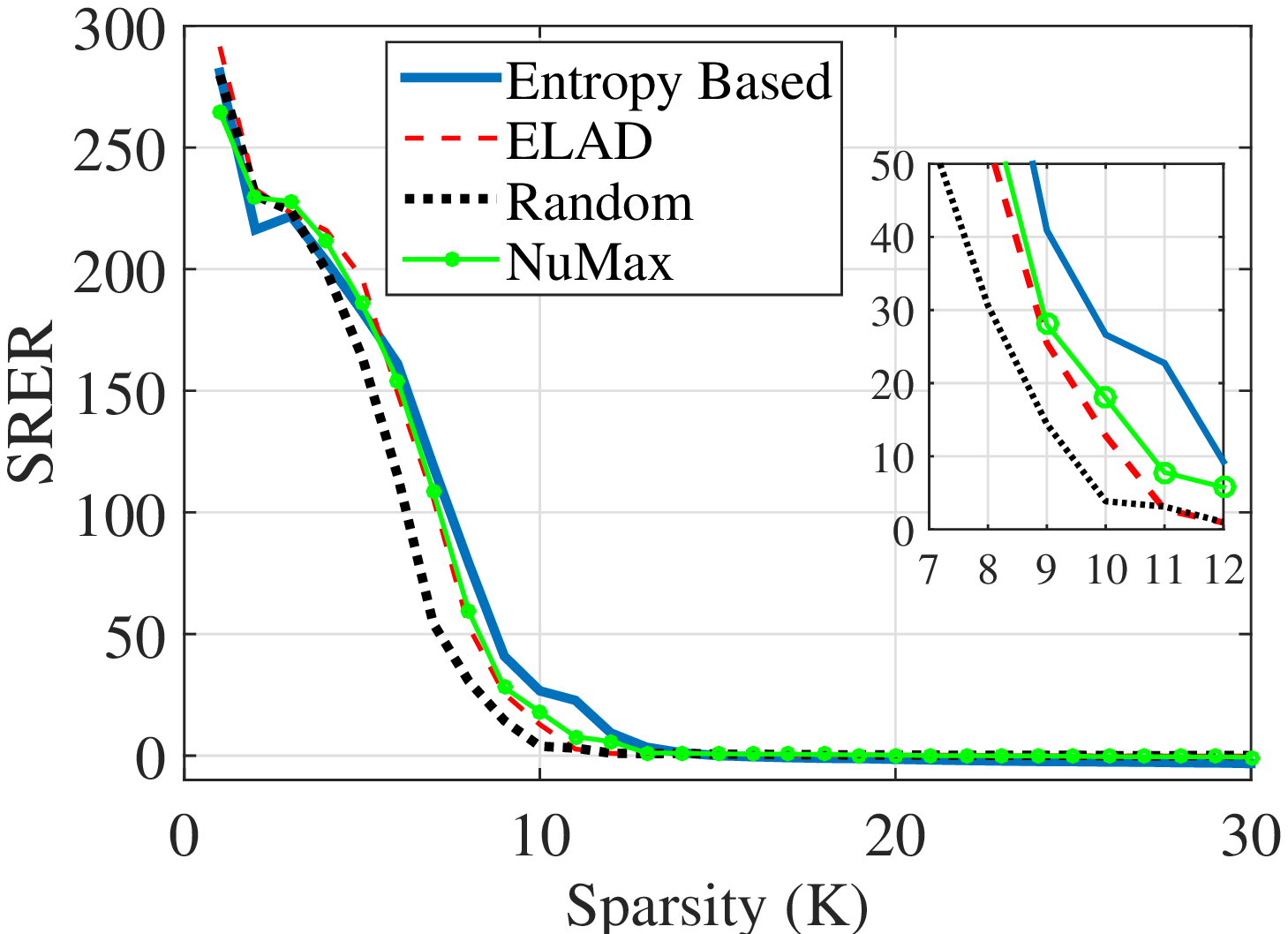}}\quad
\subfloat[][]{\includegraphics[width=5.5cm, height=4cm]{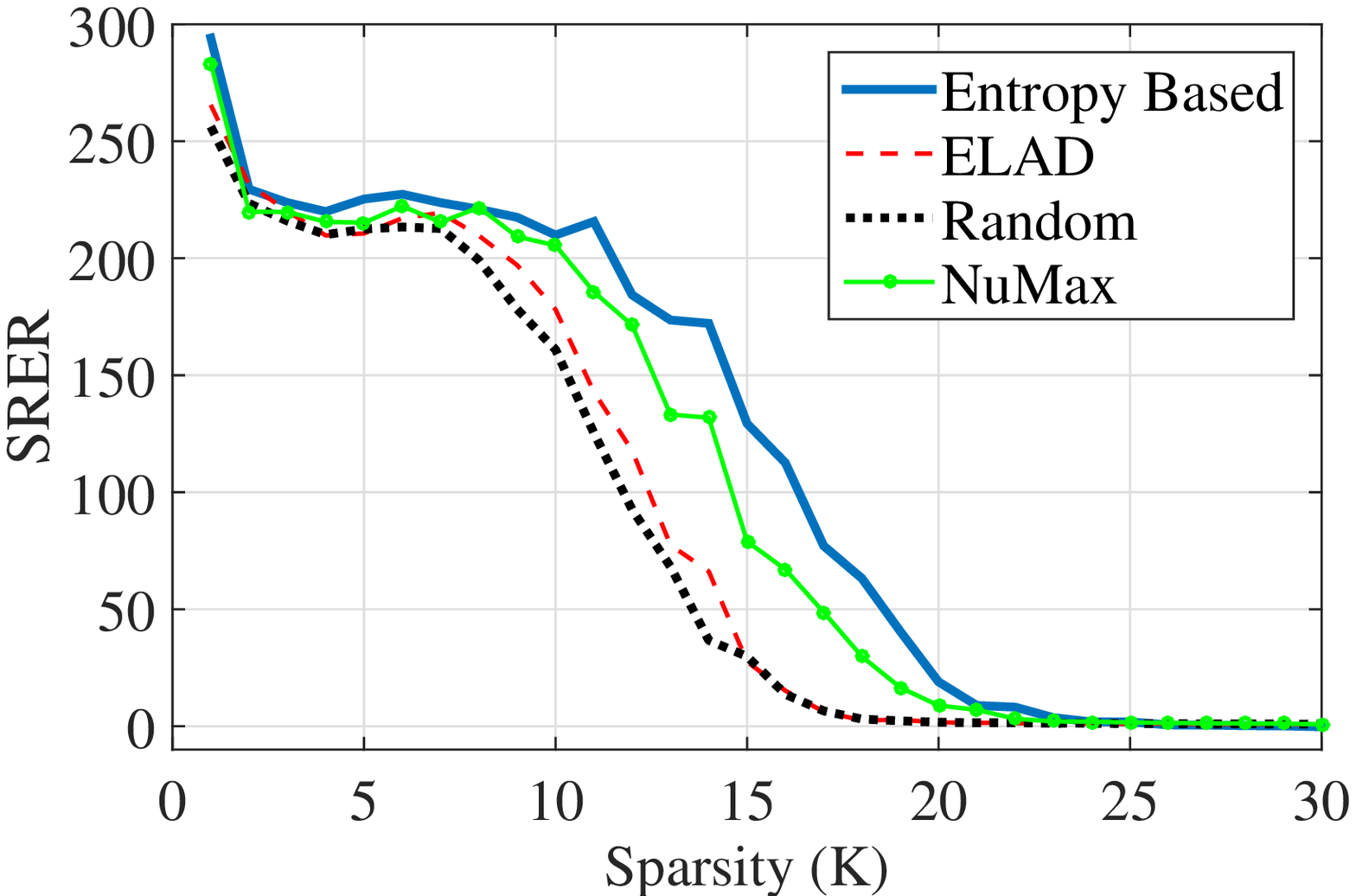}}
\caption{Average SRER (dB) of signals sparse in DCT basis recovered using BP against (a) the number of measurements, with sparsity 10 (b) sparsity, with $M=10$ (c) sparsity, with $M=20$ (d) sparsity, with $M=30$ .}
\label{reco_er_sp1}
\end{figure*}

The value of $\alpha$ in stage I decides the rate of convergence and the value of the maximum entropy to which the convergence occurs. Large values of $\alpha$ would mean that the RIP constraint is strictly followed in stage I and the maximum entropy $H(\widehat{\bm{Y}}_j)$ obtained in stage I of each iteration would be less. Since the stage II finds the matrix $\bm{A}$ such that the measurements $\bm{AC}_j$ is as close as possible to $\widehat{\bm{Y}}_j$, the final entropy $H(\bm{\widehat{A}C}_j)$ at each iteration would be in small increments. Hence, the algorithm would converge slowly to a lower value of maximum entropy. For smaller values of $\alpha$, the RIP constraint will not be followed strictly in stage I and the maximum entropy obtained in each iteration would be close to $\ln(M)$. In such a case, the algorithm will converge fast to a higher entropy value. The discussion is validated in the Fig. \ref{alpha_ver_reco}, which shows the saturation of average entropy per signal for synthetic signals for different values of $\alpha$. The experimental setup in this paper uses $\alpha = 1$. 

\subsection{Discussion}
Consider the representation $\bm{C}_j\in \bm{C}$ of the $j$-th signal $\bm{X_j}\in \bm{X}$ with respect to the basis $\bm{\Psi}$ which is measured by $\bm{A}$ to get the measurement vector $\bm{Y}_j$. Let $\bm{\Lambda}$ be the set of indices corresponding to the coefficients in $\bm{C}_j$ having high magnitudes, which capture  at least $90\%$ energy of the signal. Let $\bm{\Lambda^c}$ be the complement set containing indices corresponding to the coefficients in $\bm{C}_j$ having negligible magnitude ($<10\%$ signal energy). In a strictly sparse case, $\bm{\Lambda}$ will contain the indices corresponding to the non-zero coefficients and $\bm{\Lambda^c}$ will contain that corresponding to the  zeros. 

The algorithm identifies a matrix $\bm{A}$ such that the measurement vector of every signal in the training set has maximum entropy. The $j$-th measurement vector $\bm{Y}_j$ attains maximum entropy when its probability distribution of representation tends to be uniform; that is $|y_{ij}|\approx |y_{kj}|$ for $i\neq k$ 

Maximization of the entropy can lead to  $|y_{ij}|\approx |y_{kj}|$ for $i\neq k$ with either all the $y_{ij}$'s being small or all the $y_{ij}$'s being large. This wide separation of values may occur if the rows of $\bm{A}$ scale $\bm{C}_j$ unduly. In the proposed algorithm, since $\bm{A}$ is generated using the orthogonal Procrustes method, the rows of $\bm{A}$ are orthonormal. Hence $\bm{A}$ does not drastically scale the coefficients in $\bm{C}_j$. 

Since $\bm{A}$ does not unduly scale the coefficients in $\bm{C}_j$, the case with all the $y_{ij}$'s being small will occur when $\bm{A}$ captures the information from the coefficients present in $\bm{\Lambda^c}$, which do not contain important information of the signal. Hence this is an undesired case. The RIP penalty term ensures that the energy of the measurement vector $\bm{Y}_j$ does not deviate largely from the signal energy. Hence the case of all the $y_{ij}$'s being small is eliminated. 

Hence, the entropy maximized measurements $\widehat{\bm{Y}}_j$ will capture maximum information from the coefficients of $\bm{C}_j$ corresponding to the indices in $\bm{\Lambda}$.  

\textit{Remark:} The entropy based measure of sparsity used in this paper depends on the $2$-norm of the signal. An entropy based sparsity measure depending on $p$-norm ($p\leq 1$) is proposed in \cite{37}. Rao and Kreutz-Delgado \cite{38}  discuss the various criteria that the $p$-norm dependent entropy based measures of sparsity satisfy, for different values of $p$. The experimental results of \cite{39} show that the freedom to choose suitable values of $p$ helps in fully exploiting the sparsity-promoting potential of the entropy function. Hence, the performance of the proposed EMS algorithm may improve on choosing an appropriate sparsity promoting value of $p$ instead of the $2$-norm. We leave the detailed discussion and analysis for future study.  

\section{Results}\label{results}
This section discusses the performance of the EMS matrix constructed using the proposed algorithm applied to a class of synthetic,  speech, and image signals. The performance evaluation was done for noise-free signals and noisy signals. To measure the performance, the recovery methods used were the $l_1$ minimization or the Basis pursuit (BP) \cite{1}, the Entropy matching pursuit (EMP) \cite{17}, and the Orthogonal matching pursuit (OMP) \cite{22} algorithms. The performance evaluation is done by calculating the signal to reconstruction error (SRER) which is given by
\vspace{-5pt}
\begin{equation}
SRER=10\log_{10}\frac{\sum_i x_i^2}{\sum_i(x_i-\hat{x_i})^2},
\vspace{-2.5pt}
\end{equation}
where $\bm{x}$ is the original signal and $\hat{\bm{x}}$ is the recovered signal. Performance comparison is done for the measurements obtained from noise-free signals and noisy signals.  

\begin{figure}[!ht]
\centering
\includegraphics[width=6.0cm, height=4.4cm]{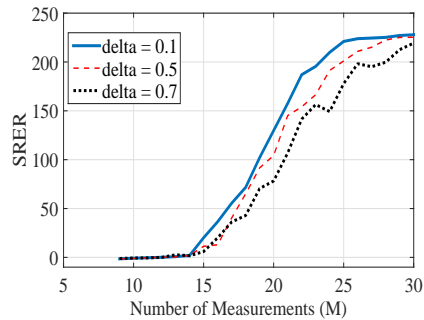}
\caption{Average SRER of signals sparse in DCT basis with sparsity 10 for $\delta=0.02, 0.1, 0.5$, and $0.7$.}
\label{del_ver_reco}
\end{figure} 
 
 The performance of the EMS matrix constructed is compared with that of the random projection matrix, the optimized measurement matrix proposed by Elad \cite{5}, and the projection matrix generated by the NuMax \cite{25}\cite{26}. The measurement matrix construction using Elad's method and the NuMax algorithm were studied experimentally, using the softwares available in \cite{23} and \cite{27}, respectively. The values of the parameters used for the construction of Elad's measurement matrix are as mentioned in \cite{5}. The training set used for the NuMax algorithm is the same as that used to train the EMS matrix.
\begin{figure*}[ht!]
\centering
 \subfloat[][]{\includegraphics[width=5.9cm, height=4.4cm]{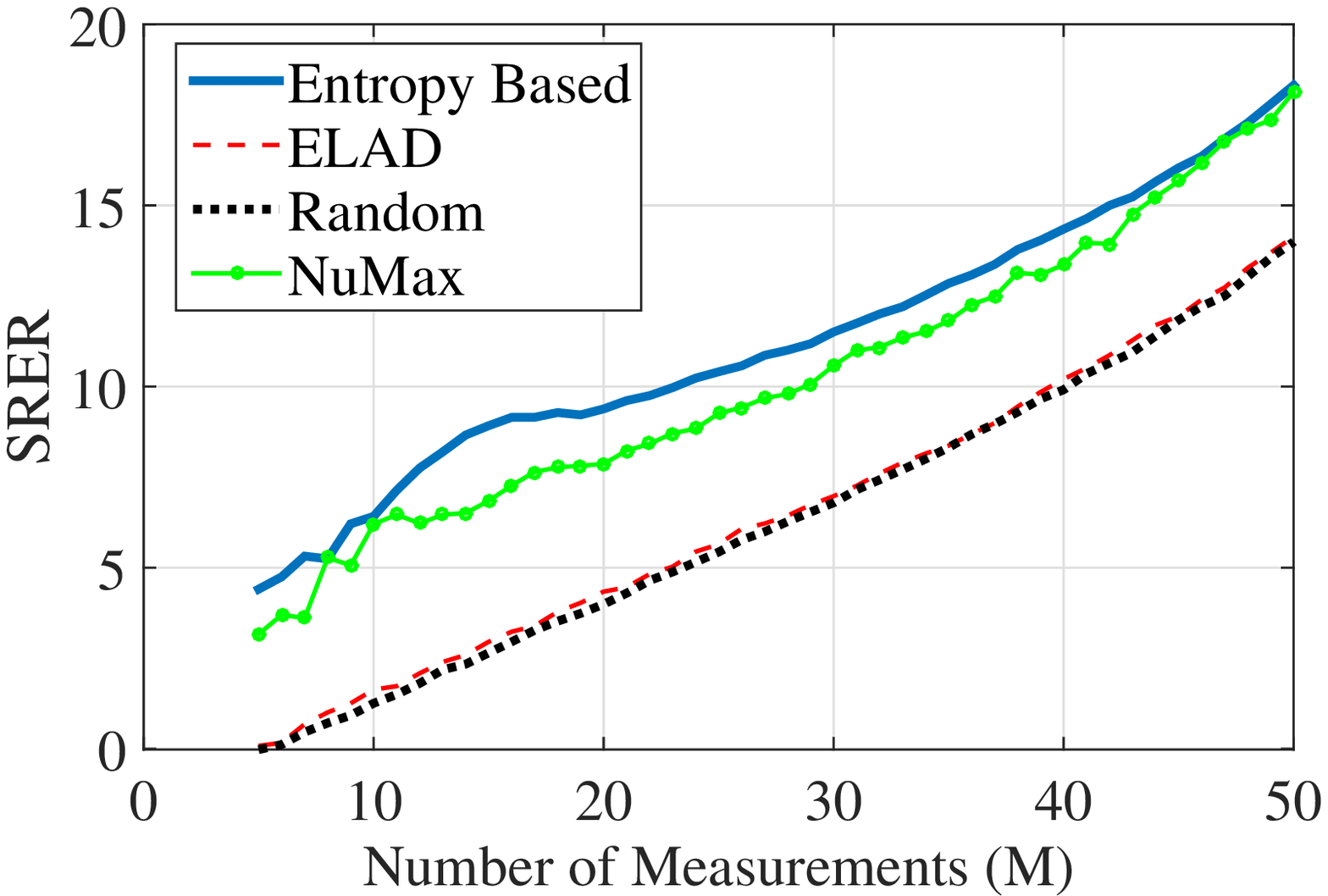}}
 \subfloat[][]{\includegraphics[width=5.9cm, height=4.4cm]{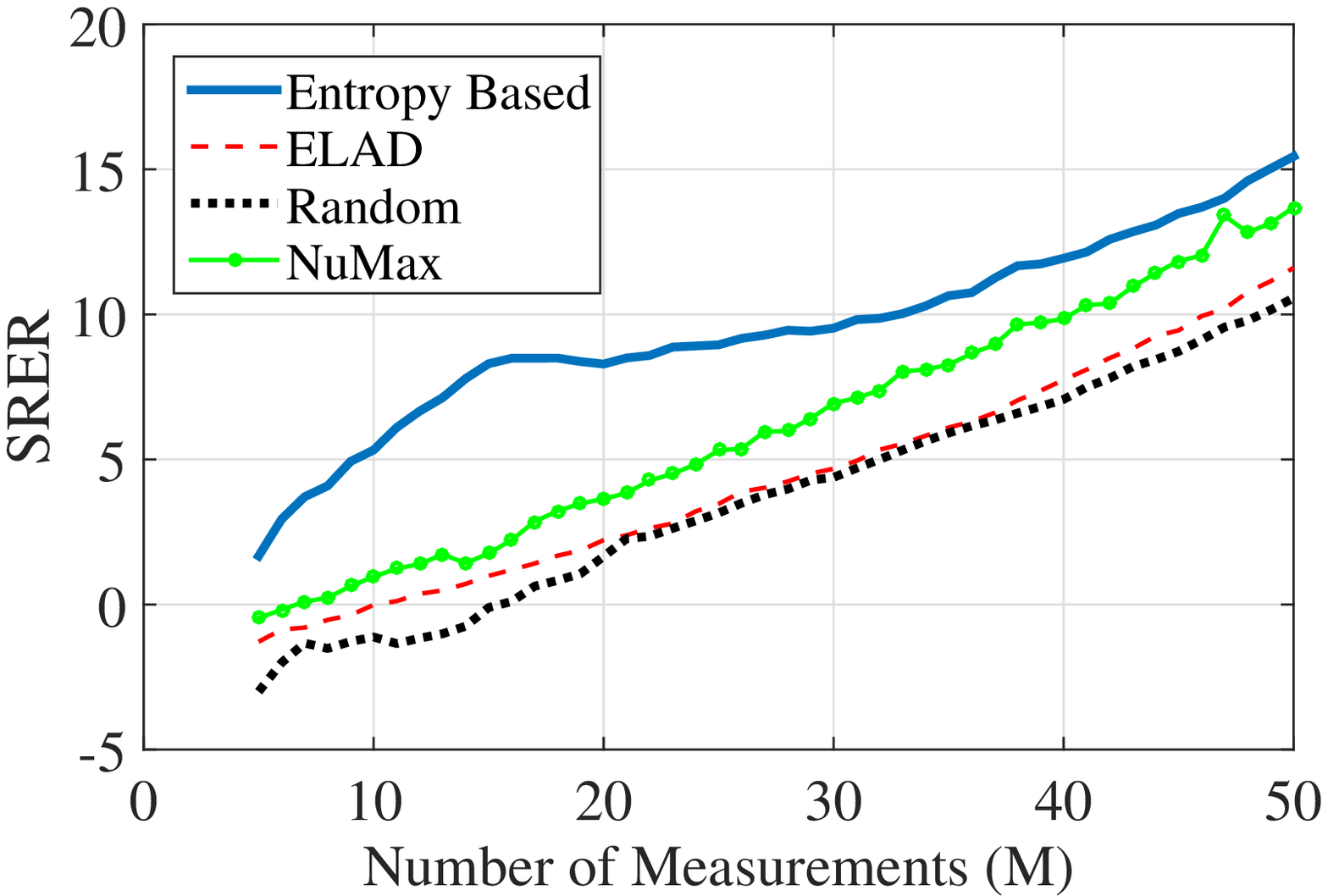}} \subfloat[][]{\includegraphics[width=5.9cm, height=4.4cm]{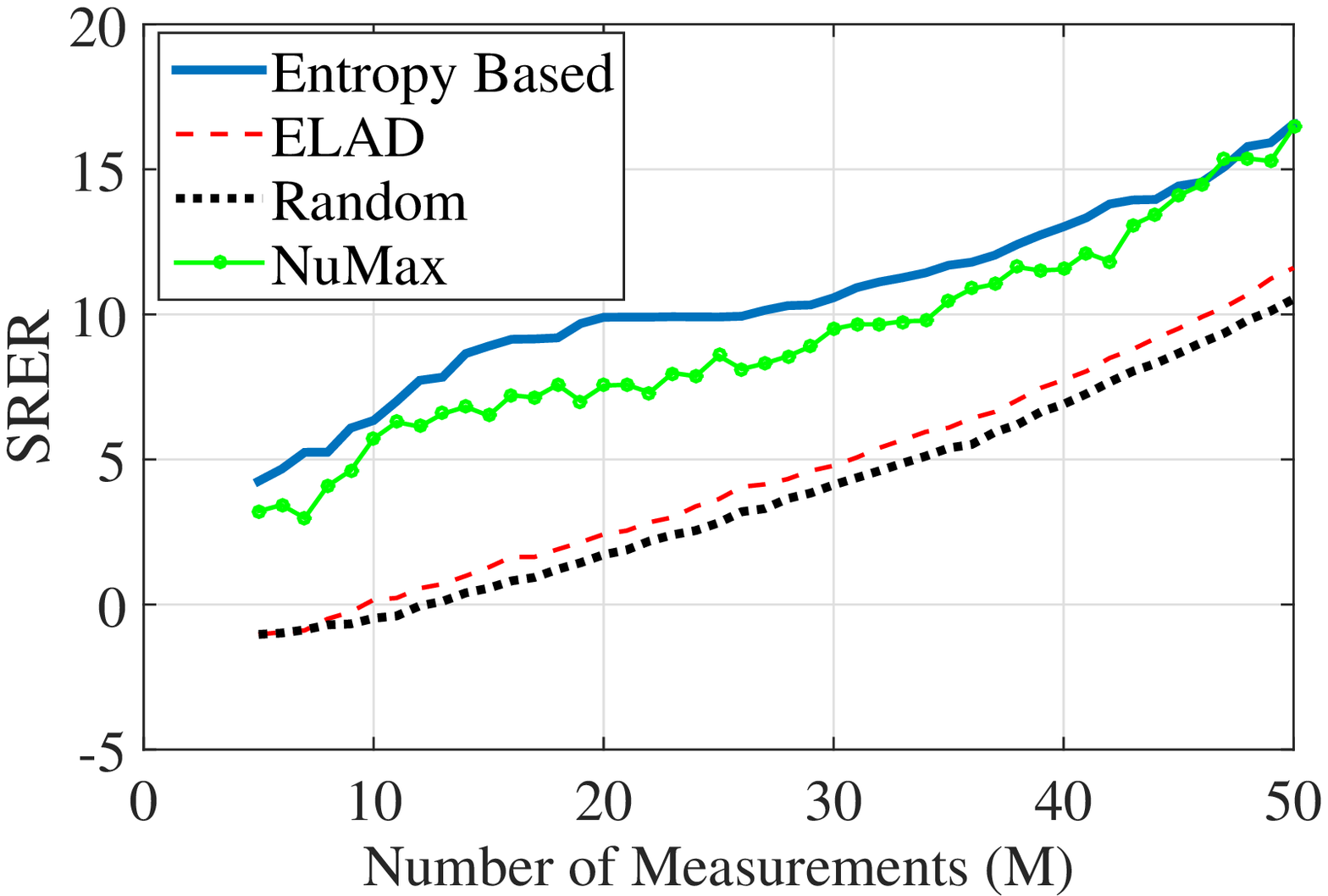}}\quad 
\caption{Average SRER (dB) of speech signals against the number of measurements with DCT as sparsifying basis and using recovery algorithms (a) BP, (b) EMP and (c) OMP.}
\label{reco_er_voice1}
\end{figure*}
\begin{figure*}[ht!]
\centering
\subfloat[][]{\includegraphics[width=5.9cm, height=4.4cm]{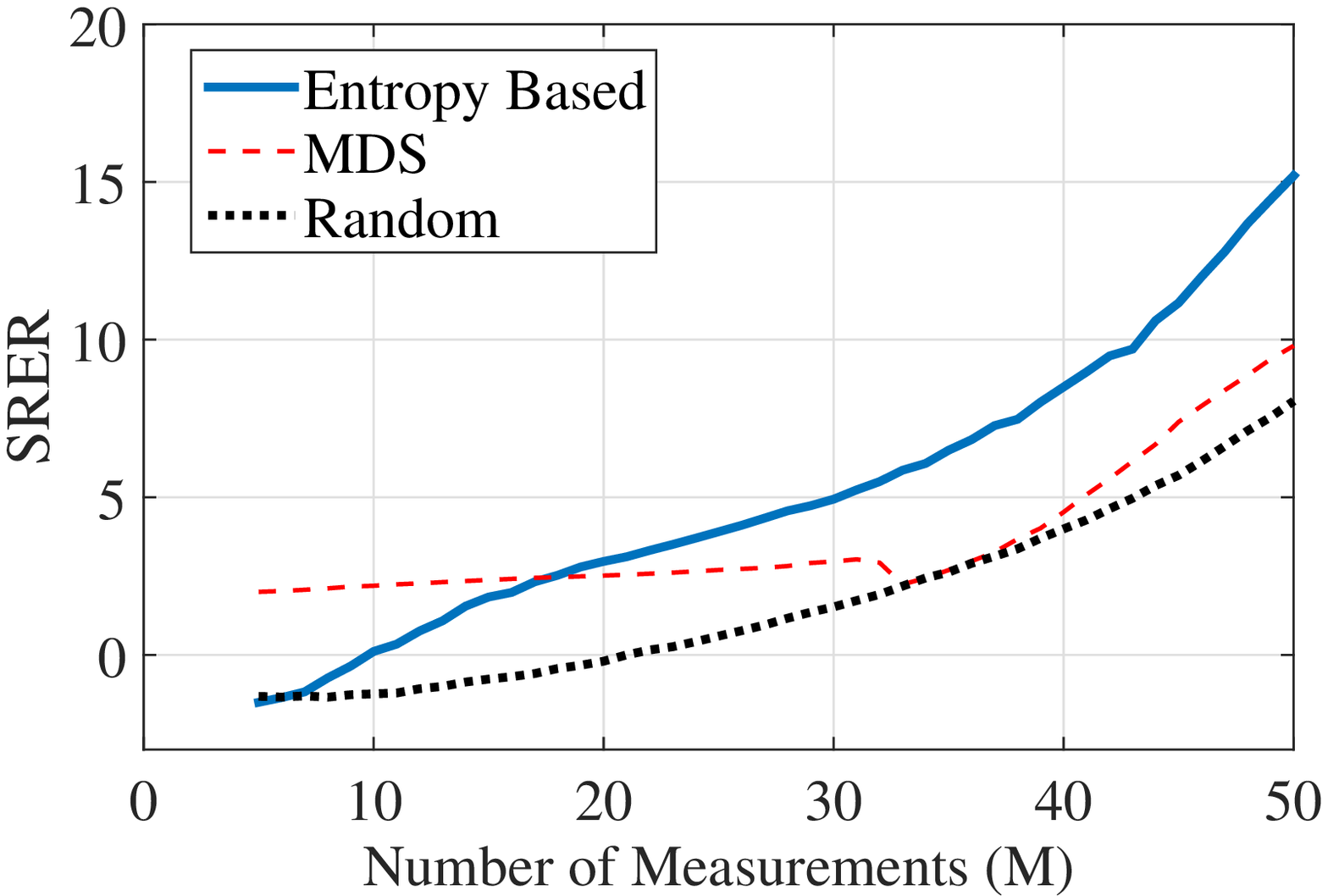}} \subfloat[][]{\includegraphics[width=5.9cm, height=4.4cm]{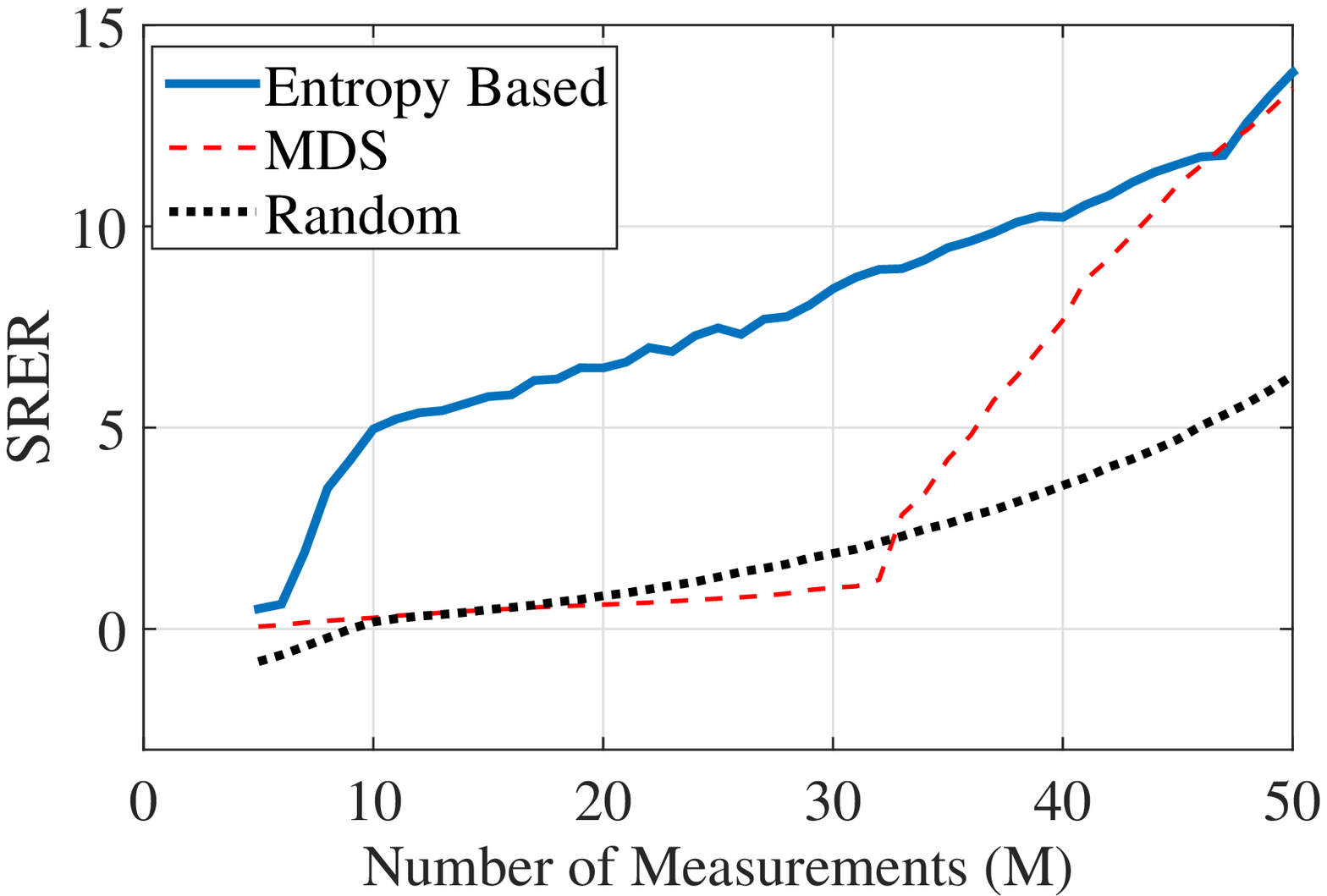}}\subfloat[][]{\includegraphics[width=5.9cm, height=4.4cm]{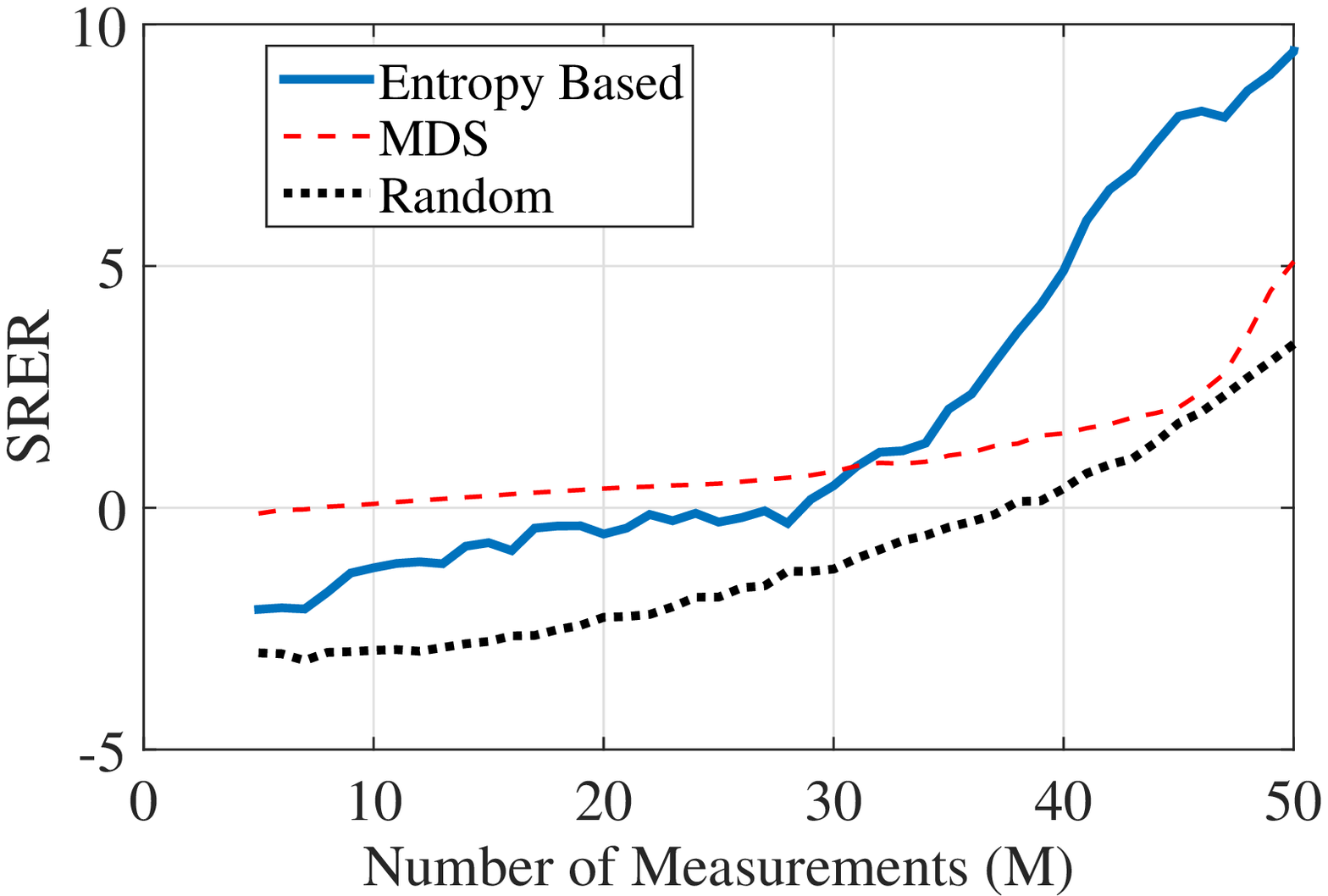}}
\caption{Average SRER (dB) of speech signals against the number of measurements with Biorthogonal wavelets as sparsifying basis and using recovery algorithms (a) BP, (b) EMP and (c) OMP.}
\label{reco_er_voice2}
\end{figure*}

\subsection{Noise-free Signals}

\subsubsection{Synthetic signals}\label{synthetic}
A set of signals, sparse in the Discrete Cosine Transform (DCT) basis with arbitrary support, was generated to test the performance of the EMS matrix. Two experiments were performed with the synthetic signals. For the first experiment, a set of $200$ signals of dimension $64$ and sparsity $10$ were generated. The measurement matrices for different values of $M$ were generated using the proposed algorithm for this class of synthetic signals. The BP algorithm was used to recover the signals from the reduced set of measurements. The average SRER plot for varying $M$ is shown in Fig \ref{reco_er_sp1}(a). 

The second experiment was to find the recovery performance when the number of measurements $M$ is fixed and the sparsity is varied. Signals of dimension $64$ and sparsity varying from $K=1$ to $30$ were constructed. A set of $200$ signals for each $K$ value was generated and concatenated to generate the training set (the training set hence contained 6000 signals). The EMS matrix was generated with $M=10$, $M=20$ and $M=30$ using these training signals. The variation of SRER against $K$ is plotted in Fig. \ref{reco_er_sp1}(b)-(d). 

The SRER plots show that the EMS matrix gives higher values of SRER even with less number of measurements, for strictly sparse signals, than attained using sensing matrices constructed using other methods. Fig. \ref{reco_er_sp1}(b) shows that for a signal with $l_0$-sparsity $K=5$ and $M=10=2K$, the SRER is close to $13$dB when the measurements were made using the EMS matrix. Whereas the SRER is less than $5$dB when other measurement matrices were used for sensing. Similarly, Fig \ref{reco_er_sp1}(c) shows that for a signal with $K=10$ the SRER is close to $26$dB when the EMS matrix with $M=2K=20$ measurements was used, which is approximately $8$dB greater than that of NuMax which gives the next best performance. Fig. \ref{reco_er_sp1}(d) shows an improvement of about $50$dB over other measurement matrices for signal with $K=15$ and $M=30$. These observations confirm the claim in Lemma \ref{lemma3}.     

The constant $\delta$ in Eqn. (\ref{max_ent3}) defines the upper bound on the radius of the recovery error sphere. It fixes the stability of a recovery algorithm which works on the measurement vector obtained using the $\bm{\Phi}$ matrix identified. The smaller the value of $\delta$ used in Eqn. (\ref{max_ent3}), the better will be the recovery as validated in Fig. \ref{del_ver_reco} which shows that as the value of delta increases, the SRER decreases.
\begin{figure*}[ht!]
\centering
 \subfloat[][]{\includegraphics[width=5.9cm, height=4.4cm]{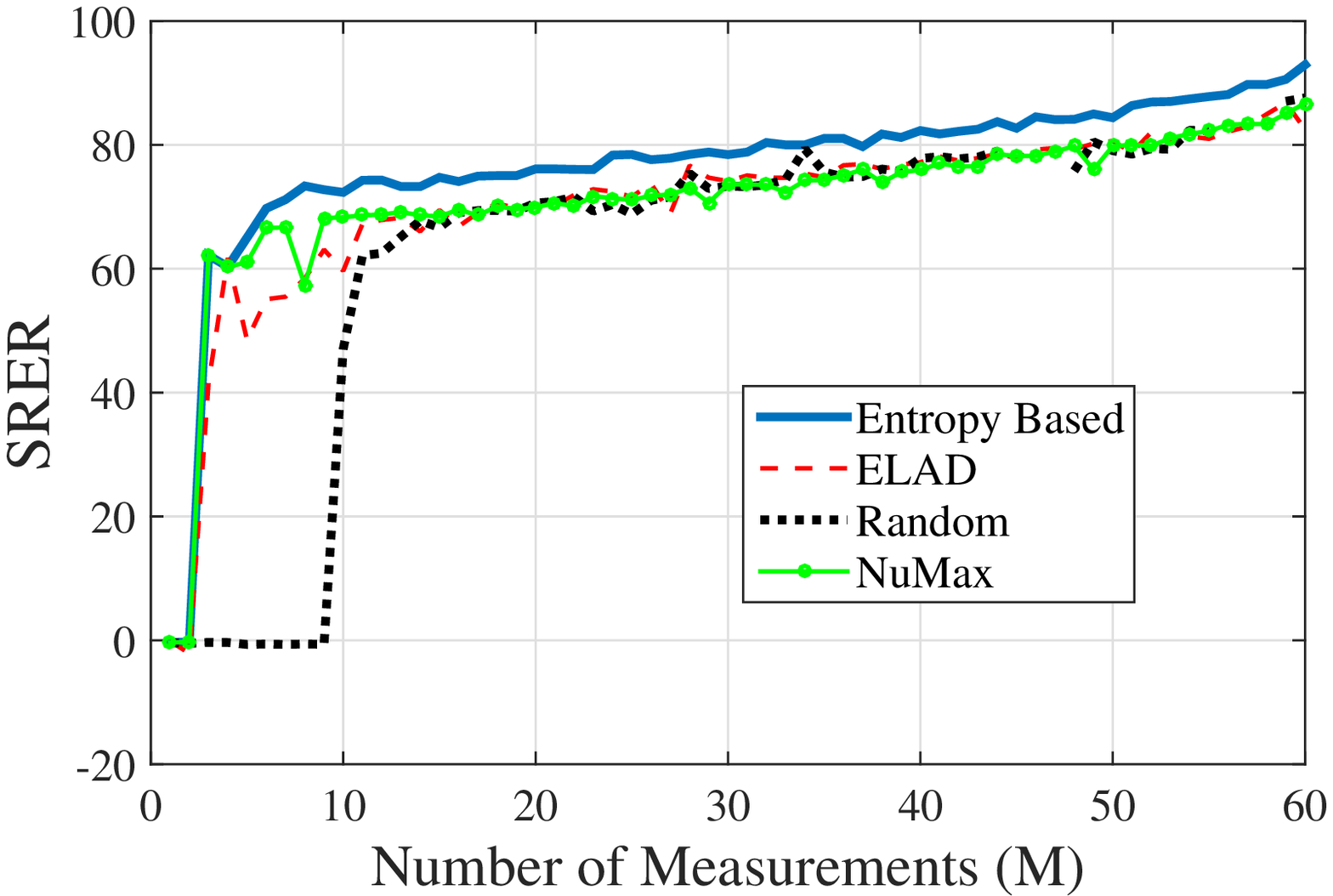}}
 \subfloat[][]{\includegraphics[width=5.9cm, height=4.4cm]{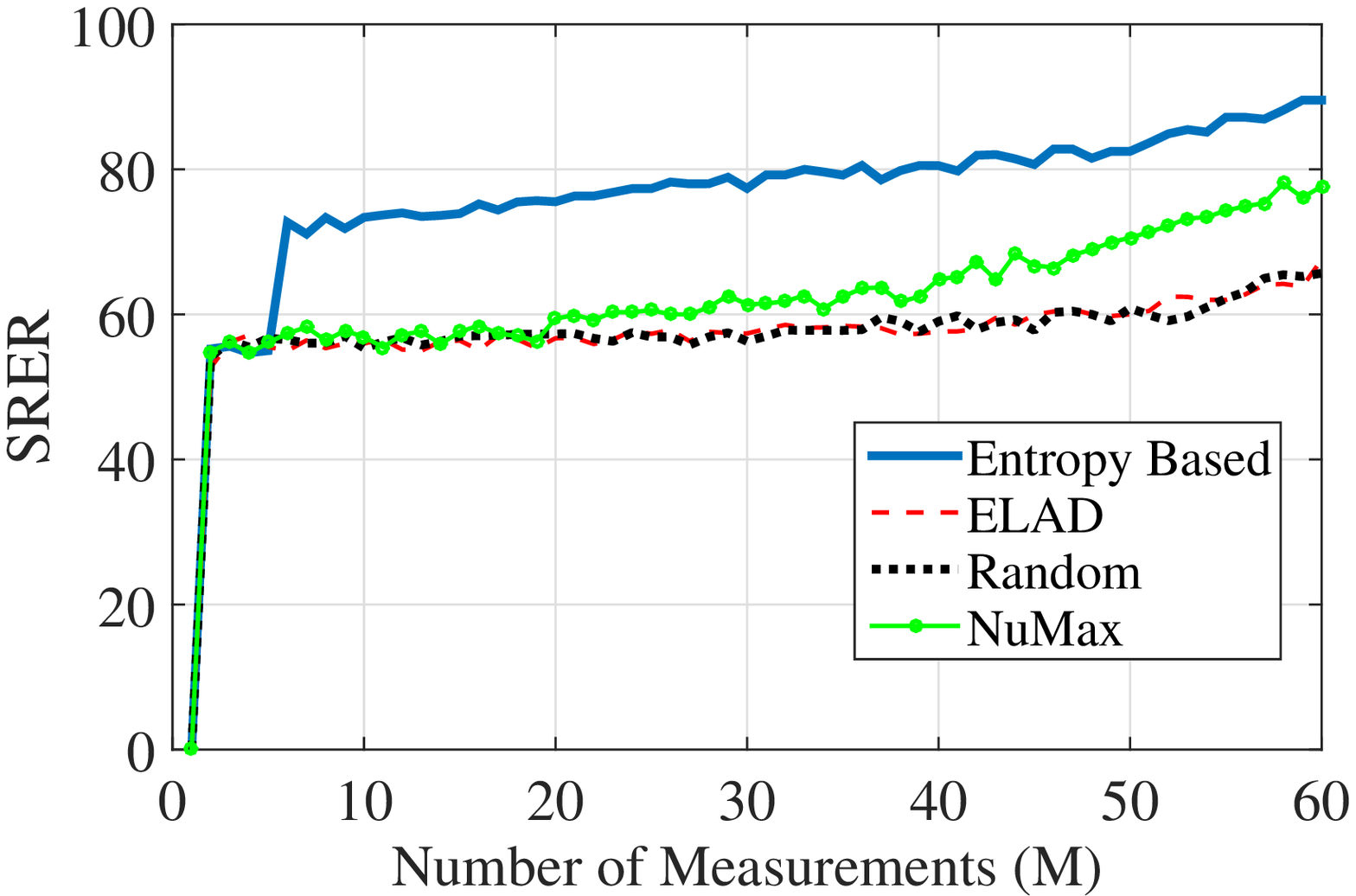}} \subfloat[][]{\includegraphics[width=5.9cm, height=4.4cm]{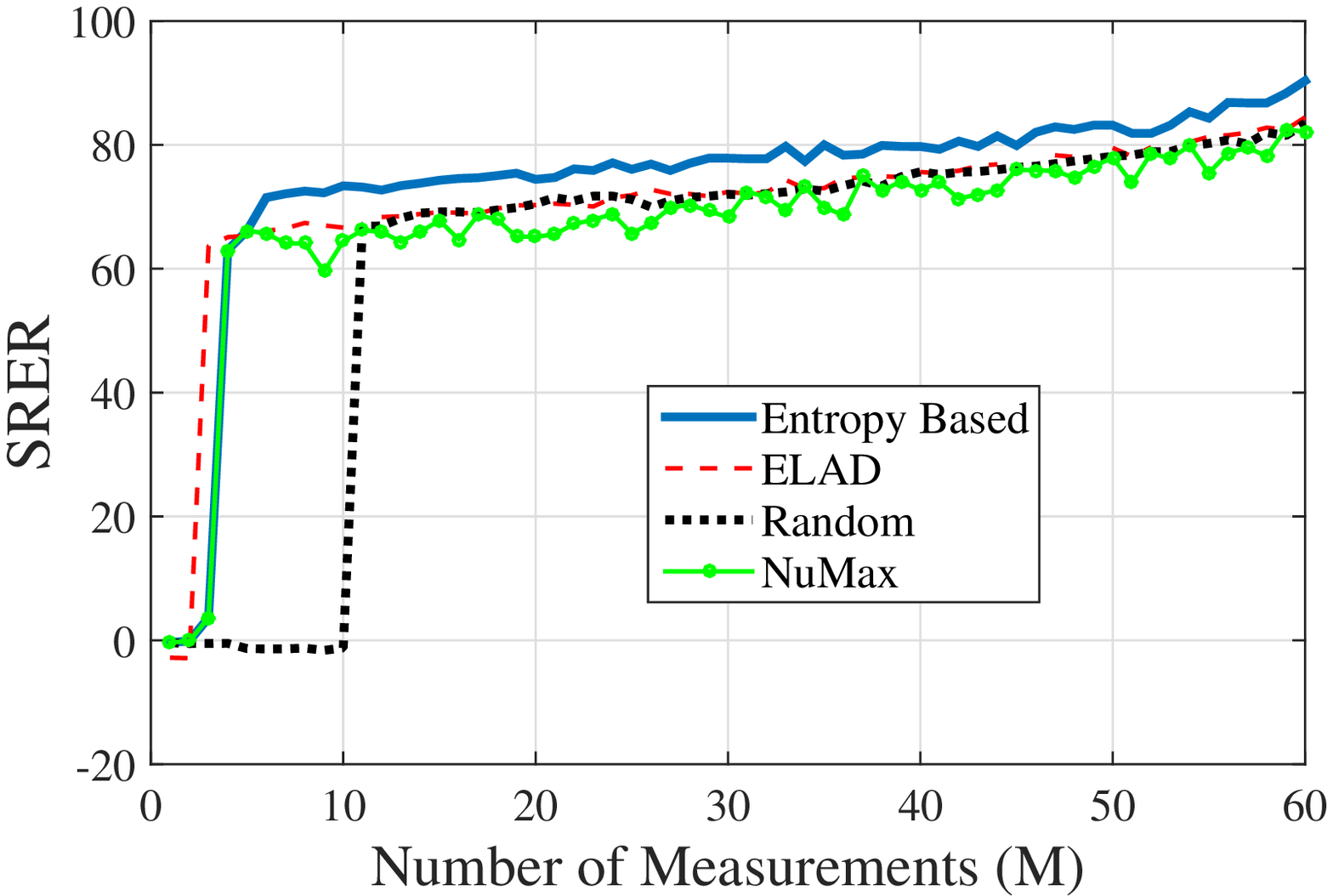}} 
\caption{Average SRER (dB) of image signals against the number of measurements with DCT as sparsifying basis and using recovery algorithms (a) BP, (b) EMP and (c) OMP.}
\label{reco_er_im}
\end{figure*}
\begin{figure*}[ht!]
\centering
\subfloat[][]{\includegraphics[width=5.25cm, height=3.75cm]{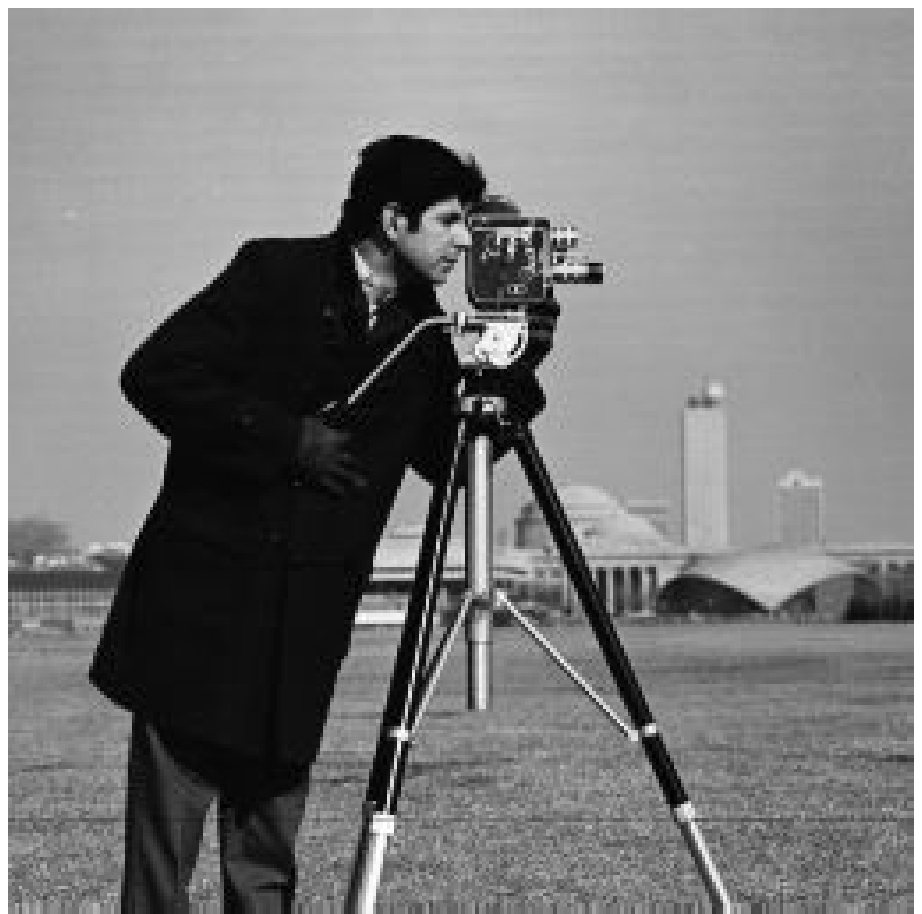}}\quad\subfloat[][]{\includegraphics[width=5.25cm, height=3.75cm]{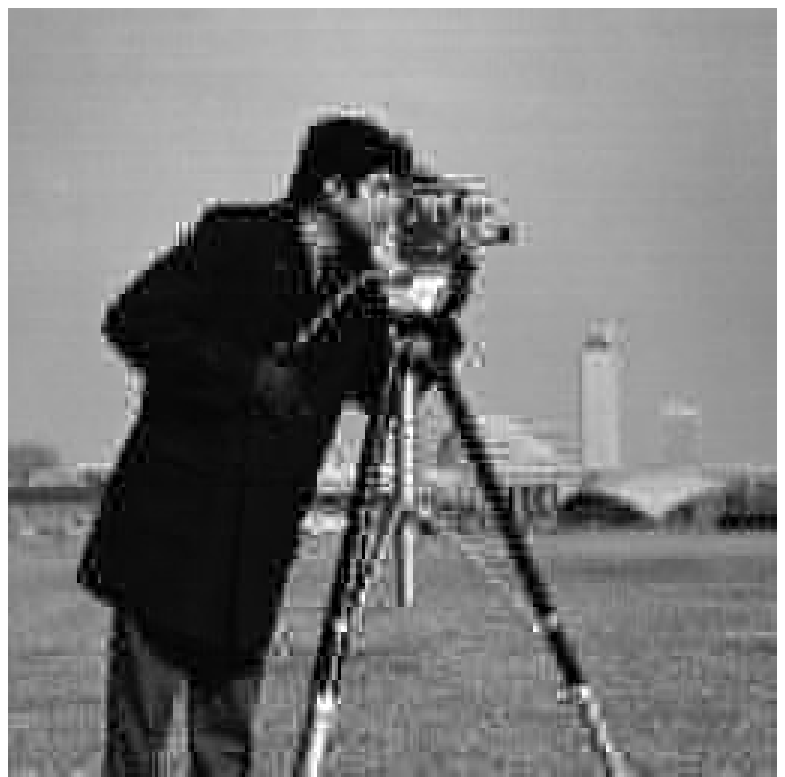}}\\ \subfloat[][]{\includegraphics[width=5.25cm, height=3.75cm]{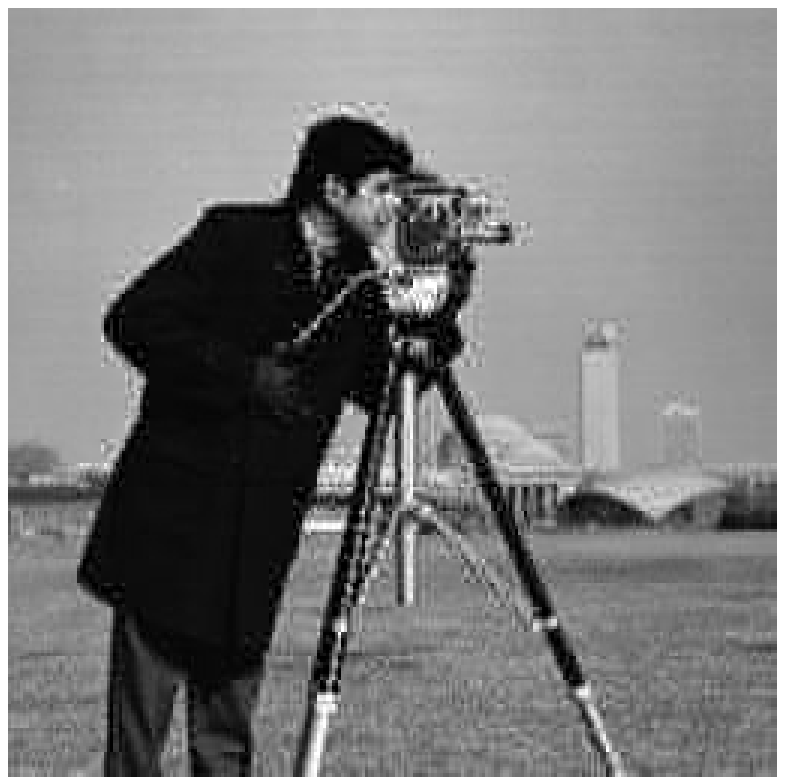}}\quad
\subfloat[][]{\includegraphics[width=5.25cm, height=3.75cm]{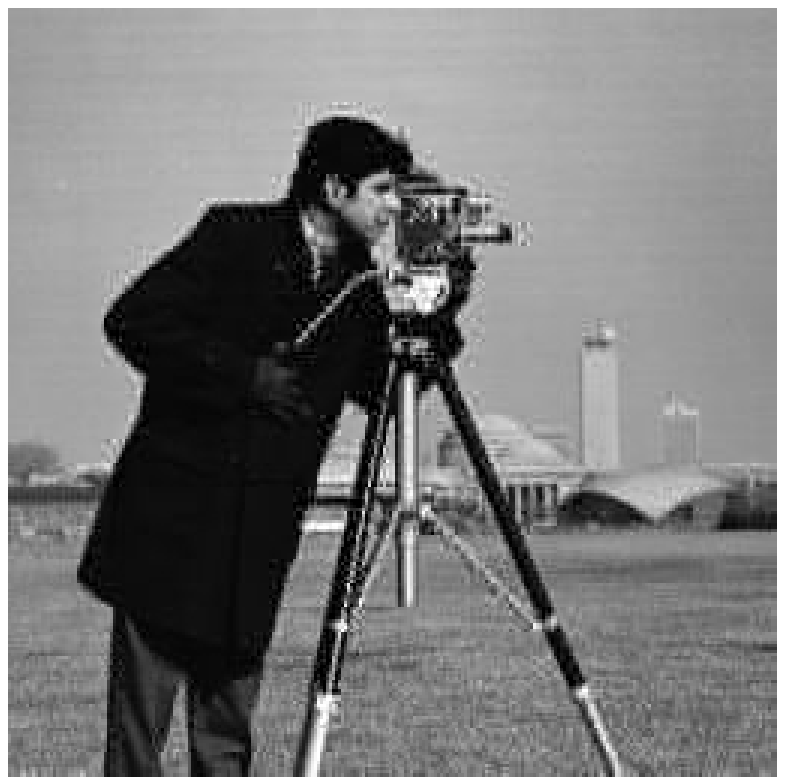}}
\caption{(a) Original image; Images measured using the EMS matrix and recovered using BP algorithm with (b) M=10 (PSNR=24.16 dB), (c) M=15 (PSNR=25.53 dB), (d) M=20 (PSNR=27.05 dB).}
\label{reco_im1.1}
\end{figure*}

The actual RIP constants obtained when the EMS matrix acts on a signal need not be the same as the one used in the algorithm. The values of $\delta$ obtained for the synthetic signals with sparsity 10 is given in the Table \ref{del_values}. The table gives the range of $\delta$ values obtained for the 200 signals. 

\begin{table}[hb!]
\centering
\caption{RIP constants for varying values of $M$ for synthetic signals with \\sparsity 10.}
\label{del_values}
{\begin{tabular*}{20pc}{@{\extracolsep{\fill}}cc@{}}\toprule
No. of measurements (M)	&Range of $\delta$\\
\hline
9						& 0.24-0.50\\
15						& 0.17-0.47\\
20						& 0.14-0.40\\
25						& 0.08-0.29\\
30						& 0.06-0.26\\
\hline
%
\end{tabular*}}{}
\end{table}  
Table \ref{del_values} shows that though the obtained values of $\delta$ are not exactly the same as the value of $\delta$ assigned in the algorithm, the measurement matrix does satisfy the RIP. Further, as the number of measurements increases, the value of $\delta$ comes closer to that of the value assigned for $\delta$.

\subsubsection{Speech signals}
A set of 2450 speech signals of dimension $64$ from the database in the Linguistic Data Consortium for Indian Languages (LDC-IL) \cite{34}, sampled at 8kHz, was used as the training set. The measurement matrices for a set of $M$ values were constructed by applying the proposed learning algorithm to these training signals with the DCT as the sparsifying basis. In Fig. \ref{reco_er_voice1}, a comparison of the SRER against $M$, for the test signals (signals belonging to the class but outside the training set) sampled using sensing matrices identified through various methods, is shown. In Fig. \ref{reco_er_voice1} the SRER curves of Elad's measurement matrix and random measurement matrix are very close. The results in \cite{5} show that the performance of the measurement matrix decreases as the $l_0$-sparsity $K$ increases. For speech signals the value of $K$ is not fixed and there may be signals in the test set with high values of $K$. This condition accounts for the low performance of the Elad's measurement matrix.  

We have compared the performance of the EMS matrix with the sensing matrix generated by applying MDS on a biorthogonal wavelet basis. For comparison, we identified the EMS matrix with the biorthogonal basis as the sparsifying basis $\bm{\Psi}$. The SRER plots are shown in Fig. \ref{reco_er_voice2}. 

The plots in Fig. \ref{reco_er_voice1} and Fig. \ref{reco_er_voice2} indicate that the sensing matrix generated leads to improved performance with the number of measurements smaller than that required by other measurement matrices. The improvement is more pronounced when EMP, an entropy based recovery algorithm, is used for recovery.

\begin{figure*}[ht!]
\centering
\subfloat[][]{\includegraphics[width=5.25cm, height=3.75cm]{cam_orig.eps}}\quad\subfloat[][]{\includegraphics[width=5.25cm, height=3.75cm]{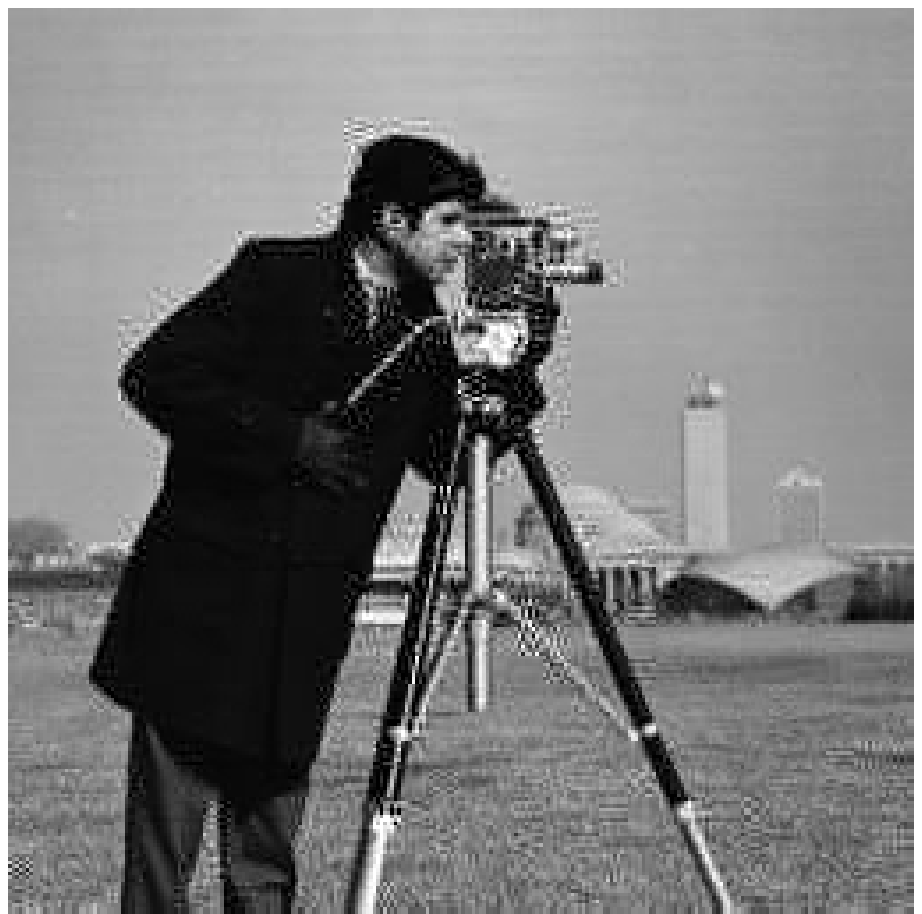}}\\ \subfloat[][]{\includegraphics[width=5.25cm, height=3.75cm]{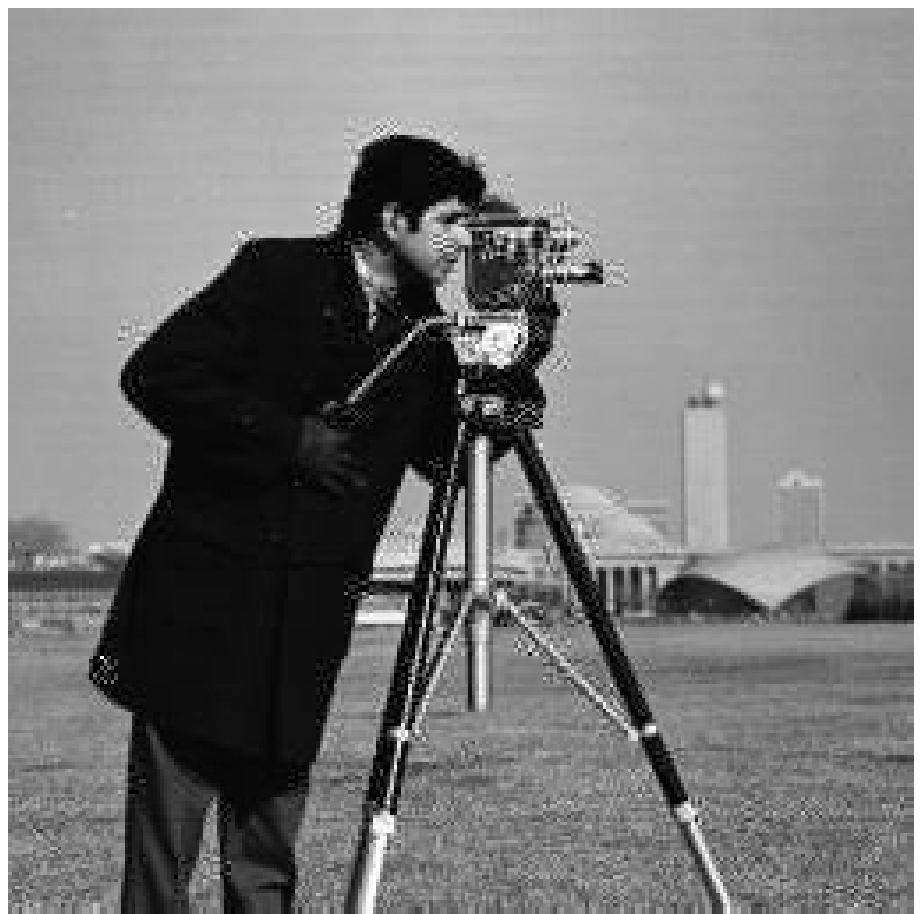}}\quad
\subfloat[][]{\includegraphics[width=5.25cm, height=3.75cm]{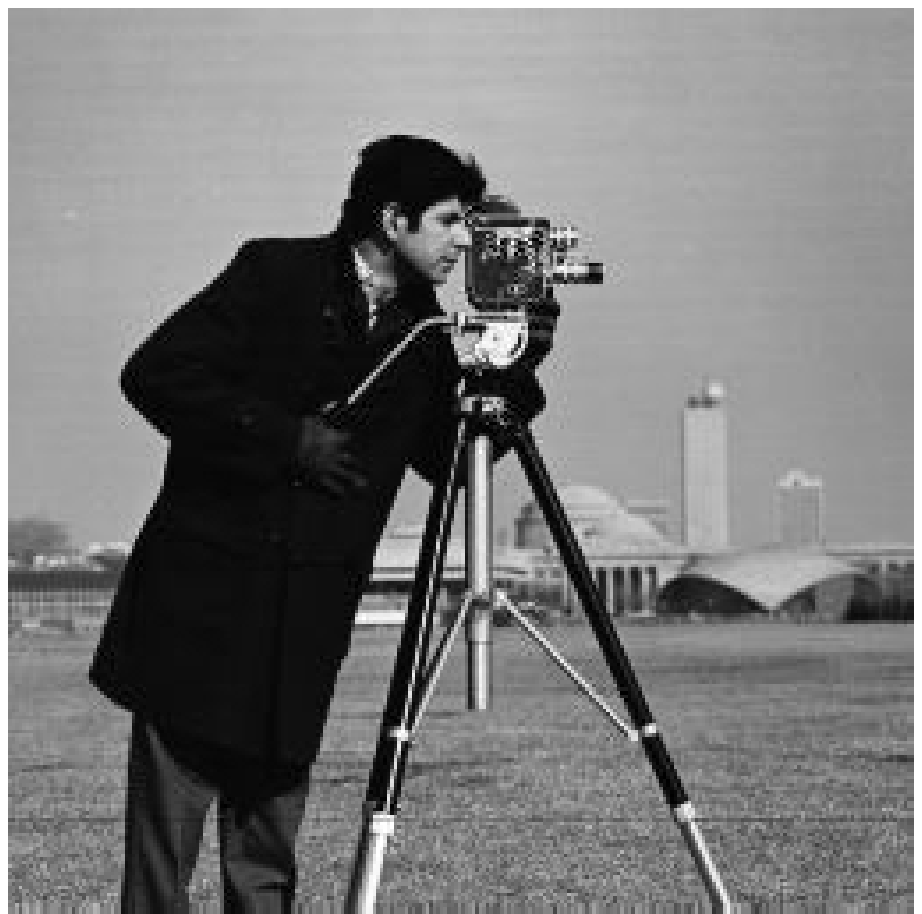}}
\caption{(a) Original image; Images measured using the EMS matrix and recovered using EMP algorithm with (b) M=25 (PSNR=25.5 dB), (c) M=35 (PSNR=27.8 dB), (d) M=55 (PSNR=37.6 dB).}
\label{reco_im1.2}
\end{figure*} 

\begin{figure*}[ht!]
\centering
\subfloat[][]{\includegraphics[width=5.25cm, height=3.75cm]{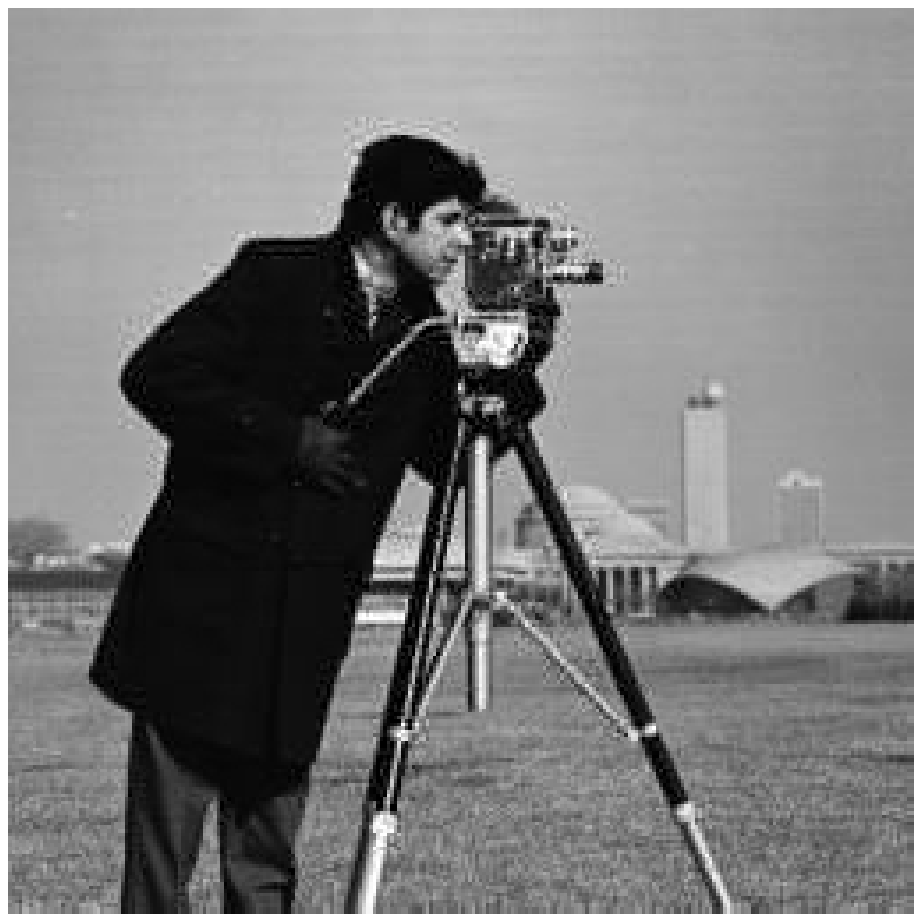}}\quad \subfloat[][]{\includegraphics[width=5.25cm, height=3.75cm]{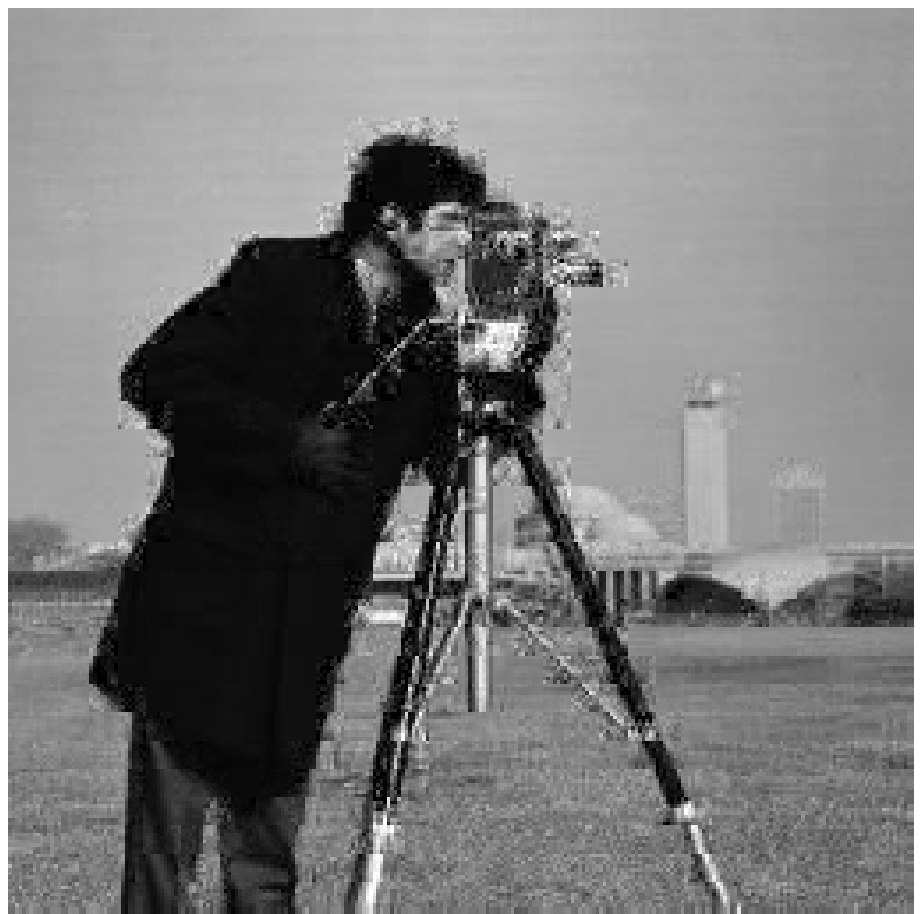}}\\
\subfloat[][]{\includegraphics[width=5.25cm, height=3.75cm]{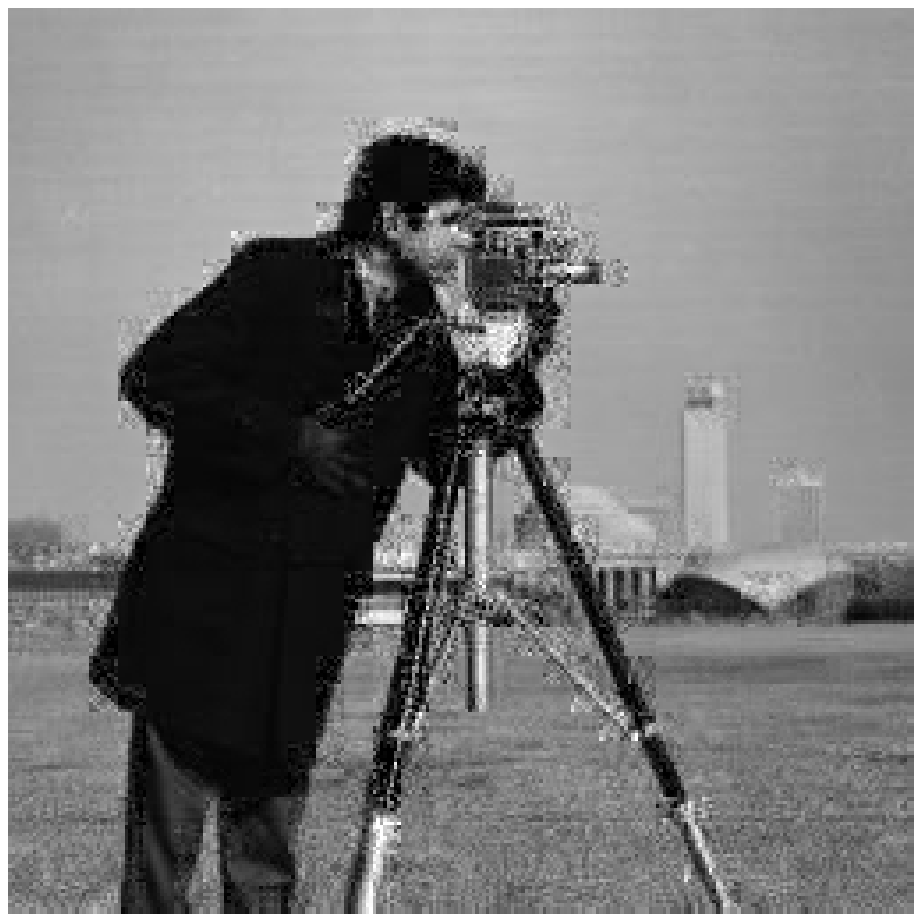}}\quad \subfloat[][]{\includegraphics[width=5.25cm, height=3.75cm]{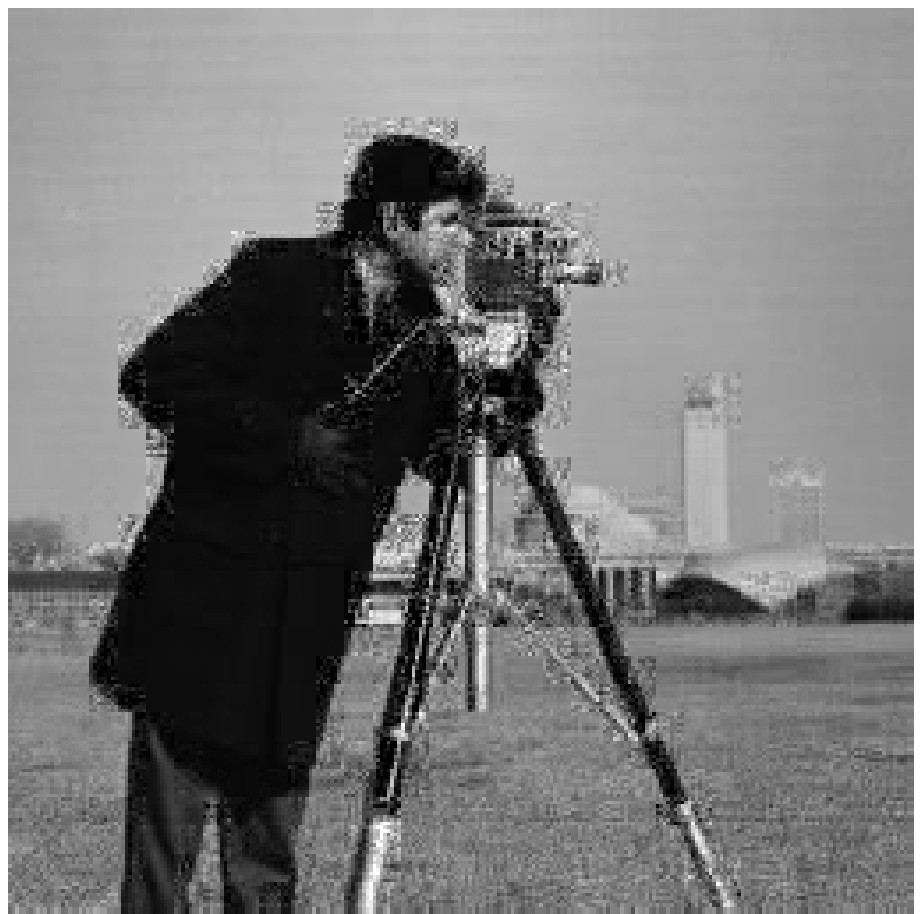}}
\caption{Images recovered using BP algorithm measured using matrices (with $M=30$) generated by algorithms (a) EMS (PSNR=29.3 dB) (b) Numax (PSNR=24.5 dB), (c) Elad (PSNR=24.2 dB) and (d) Random matrix (PSNR=23.1 dB).}
\label{reco_im2}
\end{figure*} 
\subsubsection{Image signals}
The 64-dimensional test signals for images were constructed by considering non-overlapping $8\times 8$ blocks from various images. The number of image blocks used for training were 7725. The sparsifying basis used was the DCT basis. The proposed algorithm was used to construct the measurement matrix for this class of image signals. The SRER plots in Fig. \ref{reco_er_im} show the measurement matrix constructed is able to recover signals with less number of measurements than required by other measurement matrices. 

To observe the perceptual quality of the signals, Figures \ref{reco_im1.1} and \ref{reco_im1.2}  show the images measured using the EMS matrix of various $M$ values (the $M$ value is the number of measurements in each $8\times 8$ block) and reconstructed using BP and EMP algorithms, respectively. Fig. \ref{reco_im1.1}(b)-(d) give the images reconstructed when the measurement rates ($M/N$) are 0.1, 0.2 and 0.3, respectively. The Peak Signal to Noise Ratio (PSNR) values show that the recovery error is small for measurement rate close to 0.3. Fig. \ref{reco_im2} gives a comparison of the reconstructed images when the images were measured using various measurement matrices with $M=30$, for each $8\times 8$ block, and BP was used as the recovery algorithm. It can be seen that the perceptual quality of the reconstructed image when the measurement was taken using the EMS matrix is better than that of the reconstructed images when the measurements were taken using other sensing matrices. 
\begin{figure*}[ht!]
\centering
\subfloat[][]{\includegraphics[width=5.05cm, height=4.25cm]{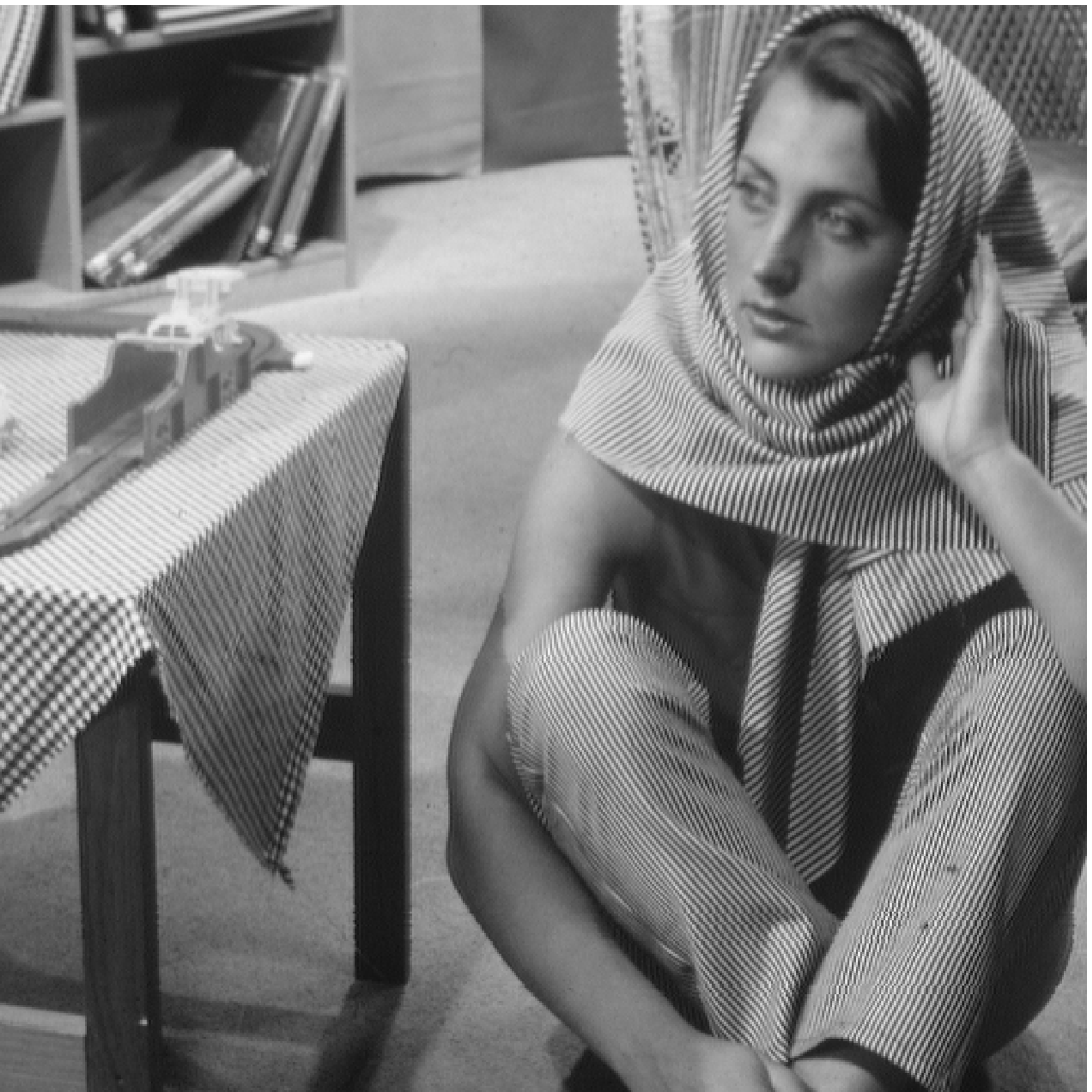}}\qquad\subfloat[][]{\includegraphics[width=5.05cm, height=4.25cm]{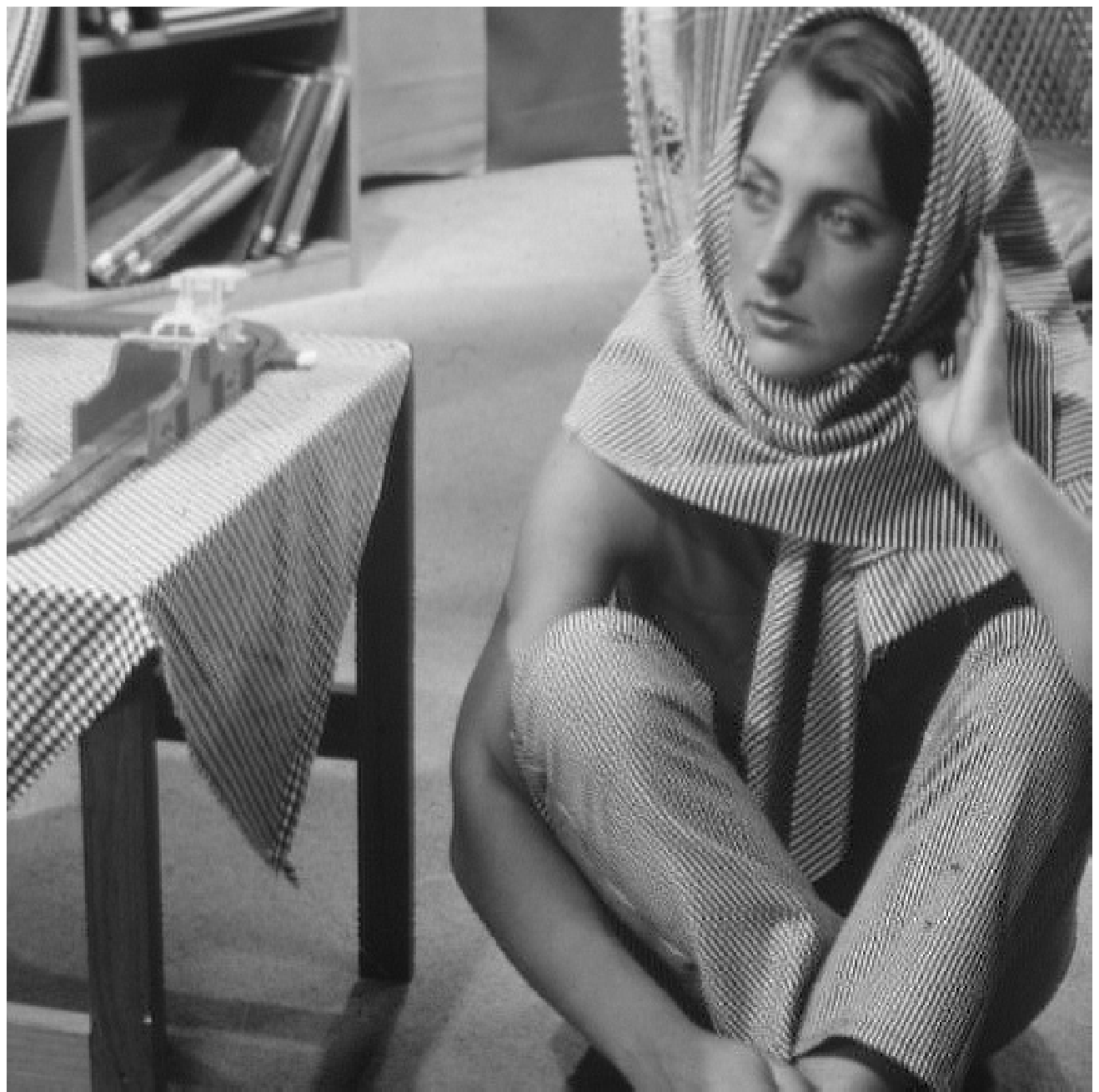}}\\ 
\subfloat[][]{\includegraphics[width=5.25cm, height=4.25cm]{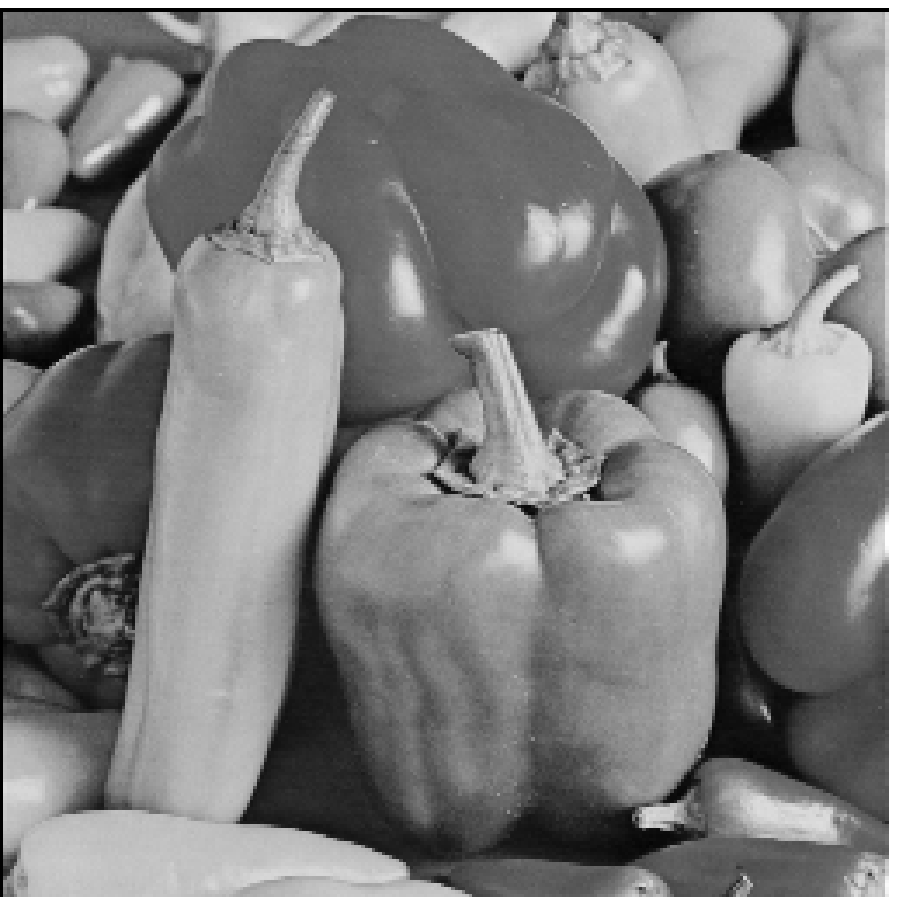}}\quad\subfloat[][]{\includegraphics[width=6.05cm, height=4.25cm]{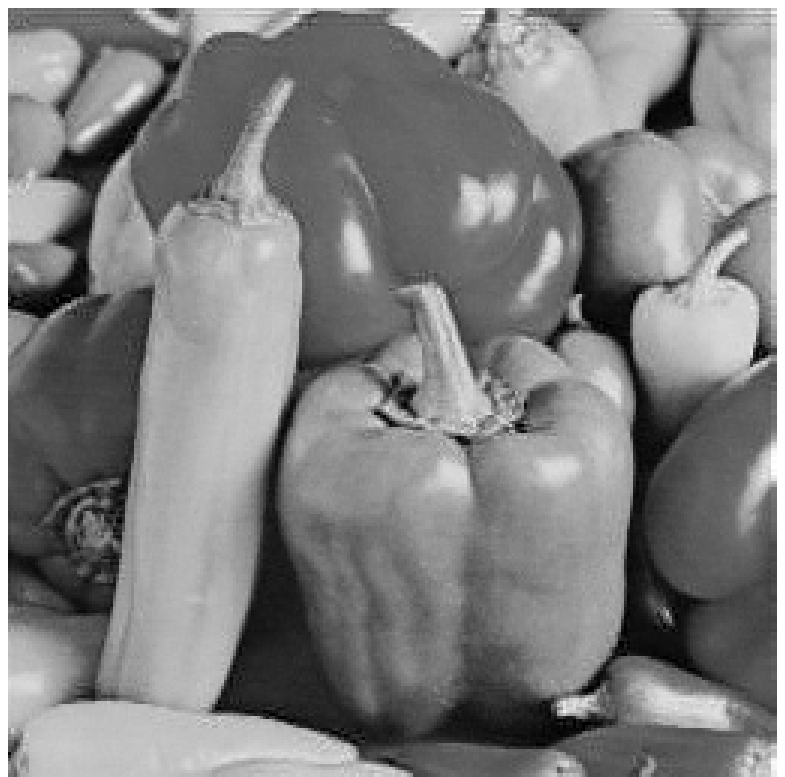}}
\caption{Images measured using the EMS matrix ($M=30$) and recovered using BP algorithm (a) Barbara original, (b) Barbara reconstructed (PSNR=29.36 dB), (c) Peppers original (d) Peppers reconstructed (PSNR=32.63 dB).}
\label{reco_im3}
\end{figure*} 

Fig. \ref{reco_im3} shows images measured using the EMS matrix with $M=30$, for each $8\times 8$ block, and reconstructed using the BP algorithm. The PSNR values show that the reconstructed images retain most of the information contained in the original image and the information loss is minimal. The figure establishes that the EMS matrix is capable of capturing maximum information from any image signal.
\begin{figure*}[ht!]
\centering
\subfloat[][]{\includegraphics[width=5.9cm, height=4.4cm]{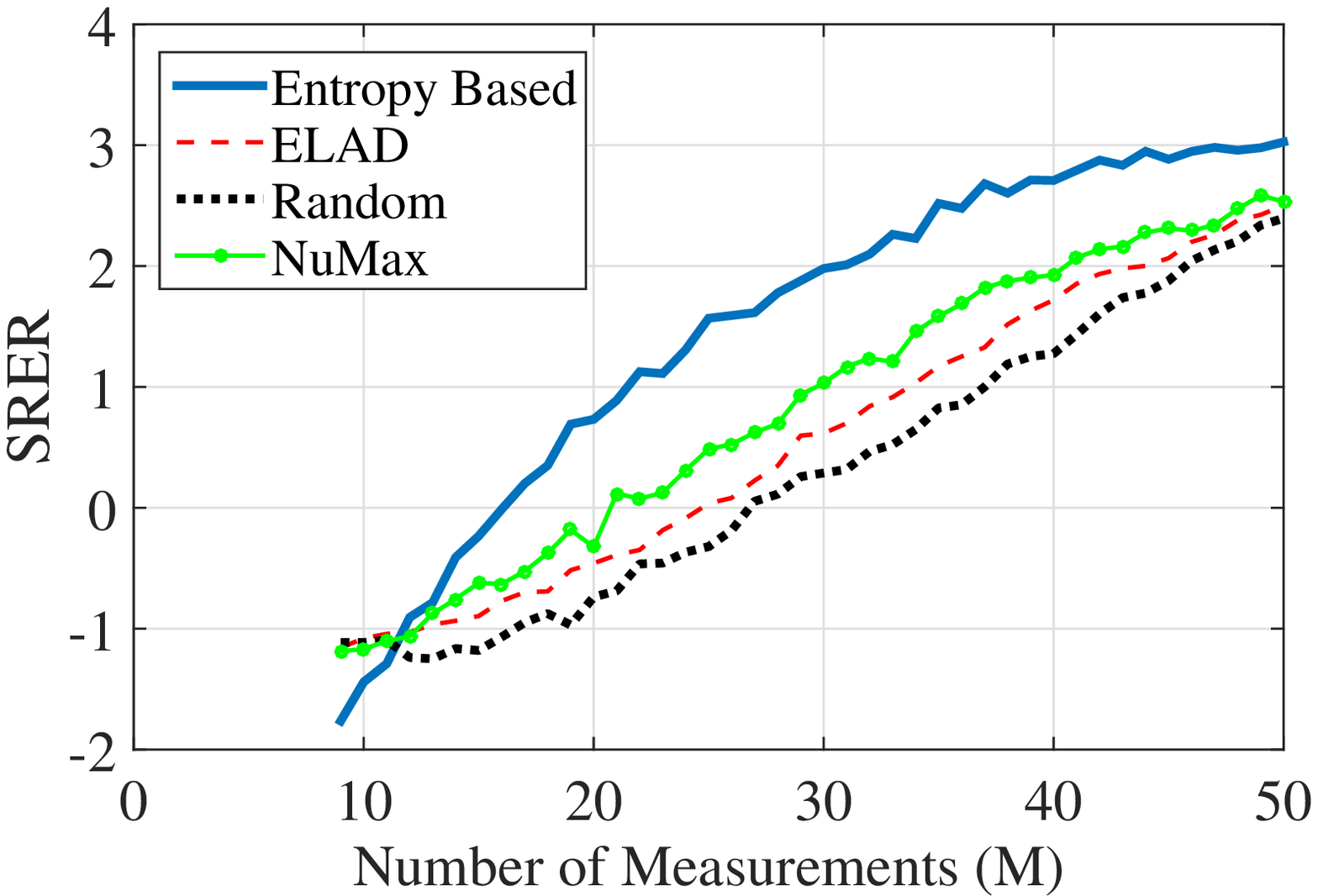}} \subfloat[][]{\includegraphics[width=5.9cm, height=4.4cm]{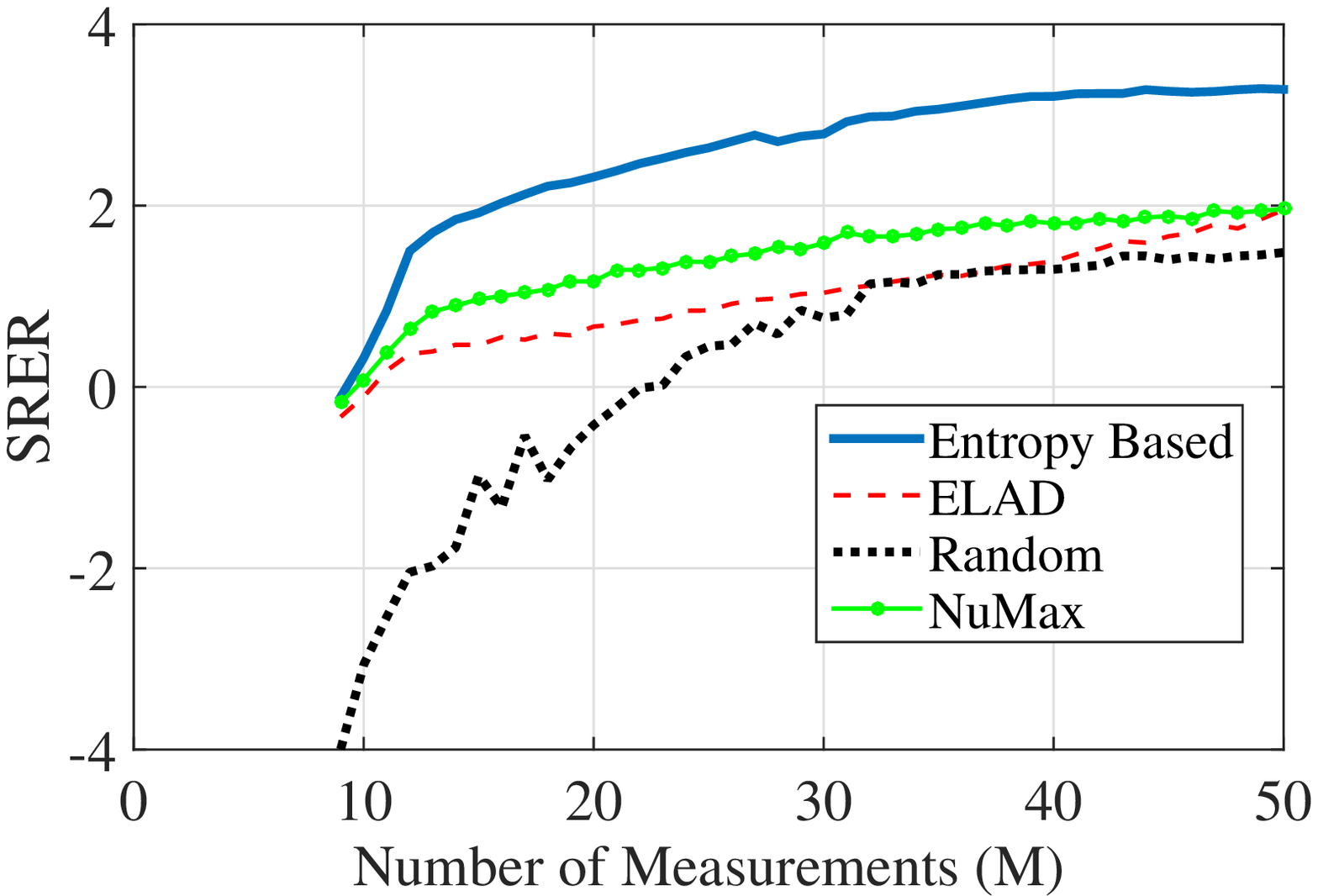}}
\subfloat[][]{\includegraphics[width=5.9cm, height=4.4cm]{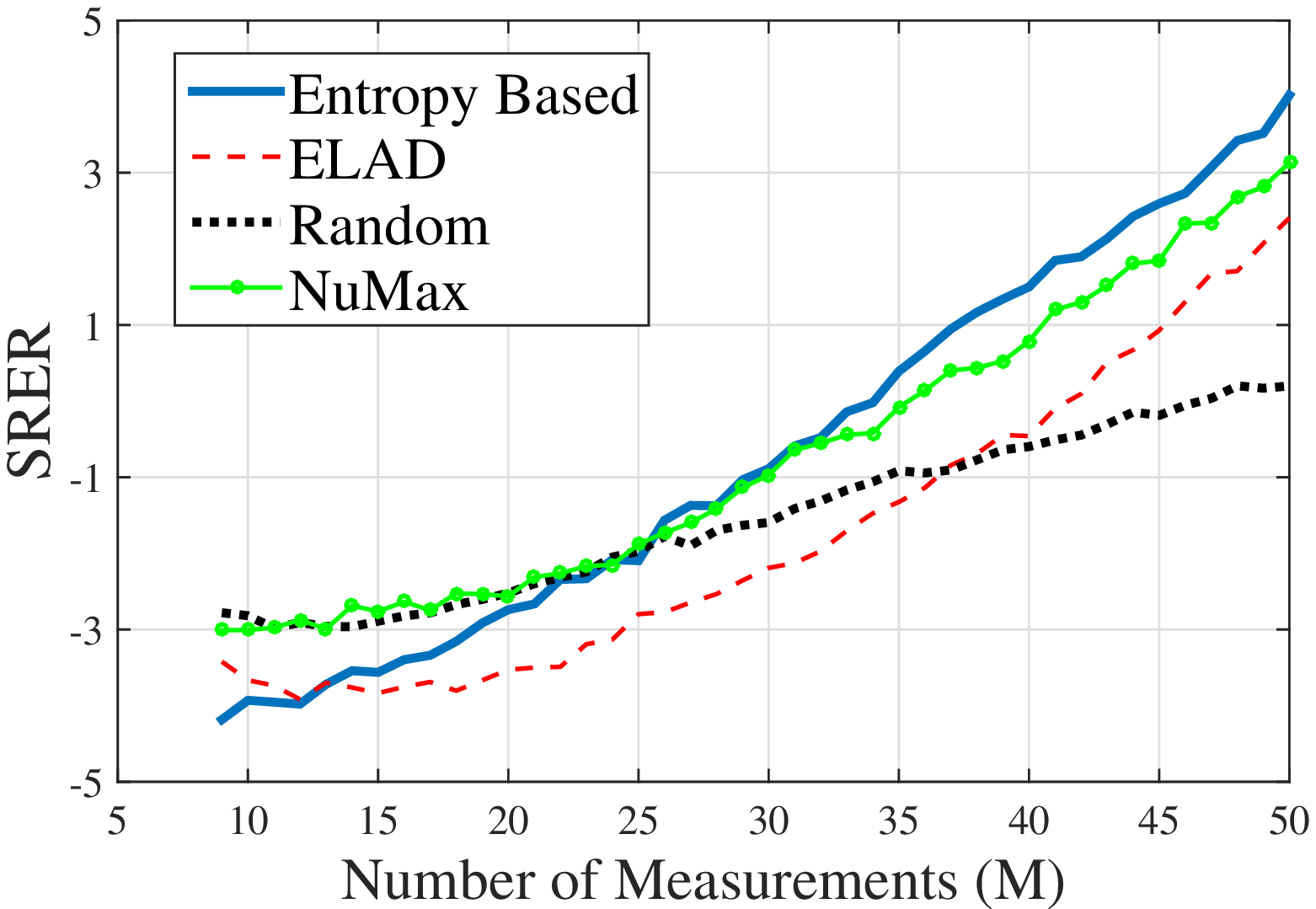}}
\caption{Average SRER (dB) of signals sparse in DCT basis with $K=10$ and input SNR 3dB and using recovery algorithms (a) BPDN, (b) EMP and (c) OMP.}
\label{noise_sparse_3db}
\end{figure*}
\subsection{Noisy Signals}
The measurement of noisy signal $\bm{x_e}=\bm{x}+\bm{e}$, where $\bm{e}$ is the additive white noise, can be expressed as
\begin{align}
\bm{y}&=\bm{\Phi x_e}=\bm{\Phi} (\bm{x}+\bm{e})\nonumber\\
&= \bm{\Phi \Psi}(\bm{c}+\bm{c_e})=\bm{A}(\bm{c}+\bm{c_e}).
\end{align}
where, $\bm{c}$ is the representation of the signal $\bm{x}$ in the basis $\bm{\Psi}$ and $\bm{c_e}$ is the representation of the noise $\bm{e}$ in the basis $\bm{\Psi}$. We know that the representation $\bm{c}$ of the signal $\bm{x}$ in $\bm{\Psi}$ is sparse but the representation $\bm{c_e}$ of the noise $\bm{e}$ is dense in $\bm{\Psi}$. The matrix $\bm{A}$ constructed using the proposed method captures maximum information from the sparse set of coefficients. Ideally, the measurement matrix $\bm{A}$ will thus capture only the noisy components having support in $\bm{\Lambda}$ and reject the noisy component having support in $\bm{\Lambda^c}$. Also, the measurement matrix is constructed such that the RIP is satisfied. The RIP ensures that the measurement matrix-recovery algorithm pair is stable. Hence, stable recovery is ensured when the EMS matrix is paired with any recovery algorithm.

To validate the performance in the presence of noise, signals contaminated with white Gaussian noise were measured using the sensing matrices and recovered using Basis Pursuit De-Noising (BPDN), EMP and OMP algorithms. The SRER plots  of signals sparse in DCT basis, with sparsity 10, for input SNR 3dB (Fig. \ref{noise_sparse_3db}) 
show that the performance of the EMS matrix paired with any recovery algorithm is stable. 

For noisy signals, the total error $(x-\hat{x})$ is contributed by the recovery error and the noise. The low SRER with less number of measurements is due to the high recovery error and the noise. As the number of measurements increases, the recovery error decreases and the major contribution to the total error is the noise. The noise reduction property of the recovery methods accounts for the marginal increase in the SRER of the recovered signals, above the input SNR, at large values of $M$. The noise reduction is more pronounced in OMP because the OMP algorithm was run for 10 iterations (since $K=10$), eliminating the noise component in the rest of the coefficients.

\section{Conclusion}
We have presented a learning method for constructing an efficient sensing matrix for the compressive sensing of a class of signals without assuming structured sparsity. The construction of the measurement matrix employs a learning scheme that maximizes the entropy of the measurement vectors of a set of training signals. We have established the bounds on the entropy of measurements necessary for the unique recovery of a signal. We have also proved that maximization of the entropy of measurements leads to a reduction in the number of measurements required for the unique recovery of signals. The sensing matrix designed was used for the compressive measurements of a class of sparse synthetic signals of arbitrary support, speech signals and image signals. The recovery of the signals is significantly better, with less number of measurements, than the recovery from the measurements obtained using other existing sensing matrices. 


\section*{Appendix}

\subsection*{A. Necessary condition for unique recovery is $M_eff\geq 2K$.}
Let $\bm{y}$ be the measurement vector corresponding to the sparse representation vector $\bm{c}$ of a compressible signal obtained as $\bm{y}=\bm{Ac}$, where the matrix $\bm{A}$ has orthogonal rows. Consider the case where $\bm{y}$ is strictly sparse with \mbox{$\Vert \bm{y} \Vert_0=M_{eff}<M$}. Here, $\bm{y}=\bm{Ac}=\tilde{\bm{A}}\bm{c}$, where $\tilde{\bm{A}}$ is formed by replacing the rows of $\bm{A}$ corresponding to the zero entries in $\bm{y}$ with rows of zeros (there would be $M-M_{eff}$ zero rows in $\tilde{\bm{A}}$). Thus, 
\begin{align*}
\text{rank}(\tilde{\bm{A}})=M_{eff}\\
\text{spark}(\tilde{\bm{A}})\in [2,M_{eff}+1].
\end{align*} 

According to the theory of CS, the condition for unique recovery is $\text{spark}(\tilde{\bm{A}}) > 2K$ [6]. 
\begin{align*}
&~~~~~~ 2K<\text{spark}(\tilde{\bm{A}})\leq M_{eff}+1\\
&\Rightarrow 2K<M_{eff}+1\\
&\Rightarrow 2K\leq M_{eff}.
\end{align*}
Thus we can say that the condition for unique recovery is $M_{eff}\geq 2K$. The argument can be extended to the case of a compressible $\bm{y}$ where $n_{th,y}^I=\lceil \exp(H(\bm{y}))\rceil=M_{eff}$. 

\subsection*{B. Solution to Orthogonal Procrustes Problem with rectangular matrix.}
If the singular value decomposition of $\bm{C}\widehat{\bm{Y}}^T$ is given by $\bm{U\Delta V}^T$, where $\bm{U}\in \mathbb{R}^{N\times N}$ and $\bm{V}\in \mathbb{R}^{M\times M}$, we propose that the solution to the problem (\ref{phi_update1}) is given by $\widehat{\bm{A}}=\bm{VU_M}^T$, where $U_M$ contains the first $M$ columns of $U$ which correspond to the largest $M$ singular values of $C\widehat{Y}^T$.
\begin{proof}
We know, $\widehat{\bm{Y}}\in \mathbb{R}^{M\times L}$, $\bm{A}\in \mathbb{R}^{M\times N}$ and $\bm{C} \in \mathbb{R}^{N\times L}$. Assuming $\bm{A}$ to have orthogonal rows we get
\begin{align*}
\Vert \widehat{\bm{Y}}-\bm{AC}\Vert_F^2 &= tr\lbrace (\widehat{\bm{Y}}-\bm{AC})(\widehat{\bm{Y}}^T - \bm{C}^T\bm{A}^T)\rbrace\\
&= tr(\widehat{\bm{Y}}\widehat{\bm{Y}}^T)-2tr(\widehat{\bm{Y}}\bm{C}^T\bm{A}^T)+tr(\bm{ACC}^T\bm{A}^T)
\end{align*}
where $tr(\bm{A})$ refers to trace of the matrix $\bm{A}$. Minimizing $\Vert \widehat{\bm{Y}}-\bm{AC}\Vert_F^2$ is equivalent to maximizing $tr(\widehat{\bm{Y}}\bm{C}^T\bm{A}^T)$
\begin{align*}
tr(\widehat{\bm{Y}}\bm{C}^T\bm{A}^T) &=tr(\bm{AC}\widehat{\bm{Y}}^T)\\
&=tr(\bm{AU\Delta V}^T)
\end{align*}
where $\bm{C}\widehat{\bm{Y}}^T=\bm{U\Delta V}^T$ and $\bm{U}\in \mathbb{R}^{N\times N}$, $\bm{\Delta} \in \mathbb{R}^{N\times M}$ and $\bm{V}\in \mathbb{R}^{M\times M}$. Since $\bm{C}\widehat{\bm{Y}}^T$ is a rectangular matrix with $M<N$, the $\bm{\Delta}$ matrix is diagonal with $N-M$ rows equal to zero. Thus, $\bm{U\Delta}=\bm{U_M \Delta_M}^T$, where $\bm{U_M}$ contains the first $M$ columns of $\bm{U}$ which correspond to the largest $M$ singular values of $\bm{C}\widehat{\bm{Y}}^T$, and $\bm{\Delta_M}^T$ contains the rows of $\bm{\Delta}$ containing the non-zero $M$ singular values of $\bm{C}\widehat{\bm{Y}}^T$ ($\bm{\Delta_M}^T$ is a square diagonal matrix). 
\begin{align*}
tr(\bm{AU\Delta V}^T) &= tr(\bm{AU_M\Delta_M}^T \bm{V}^T)\\
&= tr(\bm{V}^T \bm{AU_M\Delta_M}^T)
\end{align*}

This trace will be maximum when $\bm{V}^T \bm{AU_M}=\bm{I}$ \cite{33}. Hence, $\bm{A}= \bm{V U_M}^T$ with orthogonal rows.
\end{proof}

\end{document}